\documentclass{article}

\usepackage{amsmath}
\usepackage{tikz}
\usetikzlibrary{decorations.pathreplacing,angles,quotes}
\usetikzlibrary{patterns}
\usepackage{amsthm}
\usepackage{amssymb}
\usepackage{mathrsfs}

\usepackage{enumitem}

\usepackage[utf8]{inputenc}
\usepackage[T1]{fontenc}

\usepackage{xspace}

\usepackage{xcolor}
\usepackage{hyperref} 

\newtheorem{theorem}{Theorem}[section]
\newtheorem{lemma}[theorem]{Lemma}
\newtheorem{proposition}[theorem]{Proposition}
\newtheorem{corollary}[theorem]{Corollary}
\newtheorem{remark}[theorem]{Remark}
\newtheorem{definition}[theorem]{Definition}

\newcommand\pref{\ensuremath{\mathrm{Pref}}}
\newcommand\comp{\ensuremath{\rho}}
\newcommand\compsup{\ensuremath{\comp_{\sup}}}
\newcommand\compinf{\ensuremath{\comp_{\inf}}}
\newcommand\dic{\ensuremath{\mathrm{Dic}}}
\newcommand\bin{\ensuremath{\{0,1\}}}
\newcommand\f{\ensuremath{\mathscr{F}}}
\newcommand\p{\ensuremath{\mathscr{P}}}
\newcommand\tree{\ensuremath{\mathscr{T}}}
\newcommand\ew{\ensuremath{\lambda}}
\newcommand\F{\ensuremath{\mathcal{F}}}
\newcommand\DB{\ensuremath{\mathrm{DB}}}
\newcommand\ga[2]{\ensuremath{g_{#1}^{#2}}}
\newcommand\pstar{\ensuremath{(\star)}\xspace}
\newcommand\occ{\ensuremath{\mathrm{Occ}}}
\newcommand\pr{\ensuremath{\mathrm{Pr}}}
\newcommand\nn{\ensuremath{\mathbb{N}}}

\newcommand{\footremember}[2]{
    \footnote{#2}
    \newcounter{#1}
    \setcounter{#1}{\value{footnote}}
}

\title{Lempel-Ziv: a ``one-bit catastrophe''\\ but not a tragedy}
\author{Guillaume Lagarde\footremember{a}{Univ Paris Diderot,
    Sorbonne Paris Cit\'{e}, IRIF, UMR 7089 CNRS, F-75205 Paris,
    France.  Email: \texttt{guillaume.lagarde@irif.fr}.} \and Sylvain
  Perifel\footremember{b}{Univ Paris Diderot, Sorbonne Paris
    Cit\'{e}, IRIF, UMR 7089 CNRS, F-75205 Paris, France.  Email:
    \texttt{sylvain.perifel@irif.fr}.}}

\begin{document}

\maketitle

\begin{abstract}
  The so-called ``one-bit catastrophe'' for the compression algorithm LZ'78
  asks whether the compression ratio of an infinite word can change when a
  single bit is added in front of it. We answer positively this open question
  raised by Lutz and others: we show that there exists an infinite word $w$
  such that $\compsup(w)=0$ but $\compinf(0w)>0$, where $\compsup$ and
  $\compinf$ are respectively the $\limsup$ and the $\liminf$ of the
  compression ratios $\comp$ of the prefixes (Theorem~\ref{thm:infinite}).

  To that purpose we explore the behaviour of LZ'78 on finite words and show
  the following results:
  \begin{itemize}
  \item There is a constant $C>0$ such that, for any finite word $w$ and any
    letter $a$, $\comp(aw)\leq C\sqrt{\comp(w)\log|w|}$. Thus, sufficiently
    compressible words ($\comp(w)=o(1/\log|w|)$) remain compressible with a
    letter in front (Theorem~\ref{thm:upperbound});
  \item The previous result is tight up to a multiplicative constant for any
    compression ratio $\comp(w)=O(1/\log|w|)$ (Theorem~\ref{thm:lowerbound}).
    In particular, there are infinitely many words $w$ satisfying
    $\comp(w)=O(1/\log|w|)$ but $\comp(0w)=\Omega(1)$.
  \end{itemize}
\end{abstract}

\section{Introduction}

Suppose you compressed a file using your favorite compression
algorithm, but you realize there were a typo that makes you add a
single bit to the original file. Compress it again and you get a much
larger compressed file, for a one-bit difference only between the
original files. Most compression algorithms fortunately do not have
this strange behaviour; but if your favorite compression algorithm is
called LZ'78, one of the most famous and studied of them, then this
surprising scenario might well happen… In rough terms, that is what we
show in this paper, thus closing a question advertised by Jack
Lutz under the name ``one-bit catastrophe'' and explicitly stated 
for instance in papers of Lathrop and
Strauss~\cite{LathropStrauss}, Pierce II and
Shields~\cite{Pierce2000}, as well as more recently by
L\'opez-Vald\'es~\cite{Lopez}.

\subsection*{Ziv-Lempel algorithms}

In the paper~\cite{LZ78} where they introduce their second compression
algorithm LZ'78, Ziv and Lempel analyse its performance in terms of
finite-state lossless compressors and show it achieves the best possible
compression ratio. Together with its cousin algorithm LZ'77~\cite{LZ77}, this
generic lossless compressor has paved the way to many dictionary coders, some
of them still widely used in practice today. For instance, the
\texttt{deflate} algorithm at the heart of the open source compression program
\texttt{gzip} uses a combination of LZ'77 and Huffman coding; or the image
format GIF is based on a version of LZ'78. As another example, methods for
efficient access to large compressed data on internet based on Ziv-Lempel
algorithms have been proposed~\cite{Hoobin}.

Besides its pratical interest, the algorithm LZ'78 was the starting point of a
long line of theoretical research, triggered by the optimality result among
finite-state compressors proved by Ziv and Lempel. In recent work, for
instance, a comparison of pushdown finite-state compressors and LZ'78 is made
in~\cite{MMP11}; the article~\cite{KKNPS17} studies Lempel-Ziv and Lyndon
factorisations of words; or the efficient construction of absolutely normal
numbers of~\cite{LM16} makes use of the Lempel-Ziv parsing.

Some works of bioinformatics have also focussed on Ziv-Lempel algorithms,
since their compression scheme makes use of repetitions in a sequence in a way
that proves useful to study DNA sequences (see e.g.~\cite{ZHZC09}), or to
measure the complexity of a discrete signal~\cite{AboyHAA06} for instance.

Actually, both in theory and in practice, Ziv-Lempel algorithms are
undoubtedly among the most studied compression algorithms and we have chosen
only a very limited set of references: we do not even claim to be
exhaustive in the list of fields where LZ'77 or LZ'78 play a role.

\subsection*{Robustness}

Yet, the robustness of LZ'78 remained unclear: the question of whether the
compression ratio of a sequence could vary by changing a single bit appears
already in~\cite{LathropStrauss}, where the authors also ask how LZ'78 will
perform if a bit is added in front of an optimally compressible word. Since
the Hausdorff dimension of complexity classes introduced by Lutz~\cite{Lutz03}
can be defined in terms of compression (see~\cite{LM13}), this question is
linked to finite-state and polynomial-time dimensions as~\cite{Lopez} shows. As
a practical illustration of the issue the (lack of) robustness can cause, let
us mention that the \texttt{deflate} algorithm tries several starting points
for its parsing in order to improve the compression ratio.

In this paper, we show the existence of an infinite sequence $w$ which is
compressible by LZ'78, but the addition of a single bit in front of it makes
it incompressible (the compression ratio of $0w$ is non-zero, see
Theorem~\ref{thm:infinite}), thus we settle the ``one-bit catastrophe''
question. To that end, we study the question over finite words, which enable
stating more precise results. For a word $w$ and a letter $a$, we first prove
in Theorem~\ref{thm:upperbound} that the compression ratio $\comp(aw)$ of $aw$ cannot
deviate too much from the compression ratio $\comp(w)$ of $w$:
$$\comp(aw)\leq 3\sqrt{2}\sqrt{\comp(w)\log|w|}.$$
In particular, $aw$ can only become incompressible ($\comp(aw)=\Theta(1)$) if
$w$ is already poorly compressible, namely $\comp(w)=\Omega(1/\log n)$. This
explains why the one-bit catastrophe cannot be ``a tragedy'' as we point out
in the title.

However, our results are tight up to a constant factor, as we show in
Theorem~\ref{thm:lowerbound}: there are constants $\alpha,\beta>0$ such that,
for any $l(n)\in[90^2\log^2n, \sqrt{n}]$, there are infinitely many words $w$
satisfying
$$\comp(w)\leq\alpha\frac{\log|w|}{l(|w|)}\quad\text{whereas}\quad
\comp(0w)\geq\beta\frac{\log|w|}{\sqrt{l(|w|)}}.$$
In particular, for $l(n)=90^2\log^2n$, these words satisfy
$$\comp(w)\leq\frac{1}{\log|w|}\quad\text{and}\quad
\comp(0w)\geq\frac{\beta}{90}$$
(this is the one-bit catastrophe over finite words). But actually the story
ressembles much more a tragedy for well-compressible words. Indeed, for
$l(n)=\sqrt{n}$ we obtain:
$$\comp(w)\leq\alpha\frac{\log|w|}{\sqrt{|w|}}\quad\text{whereas}\quad
\comp(0w)\geq\beta\frac{\log|w|}{|w|^{1/4}},$$
that is to say that the compression ratio of $0w$ is much worse than that of
$w$ (which in that case is optimal). To give a concrete idea, the bounds given
by our Theorem~\ref{thm:toy} for words of size 1~billion ($|w|=10^9$) yield a
compression for $w$ of size at most $d\log d\leq 960{,}000$ (where
$d=1.9\sqrt{|w|}$), whereas for $0w$ the compression size is at least
$d'\log d'\geq 3{,}800{,}000$ (where $d'=0.039|w|^{3/4}$).\footnote{Actually,
  throughout the paper we preferred readability over optimality and thus did
  not try to get the best possible constants; simulations show that there is a
  lot of room for improvement, since already for small words the difference is
  significant (using notations introduced in Sections~\ref{sec:notations}
  and~\ref{sec:toy}, for $w=\pref(x)$ with $x\in\DB(12)$, $|w|\simeq 8.10^6$
  and $w$ is parsed in about $4100$ blocks, whereas $0w$ is parsed in more
  than $200{,}000$ blocks).}

This ``catastrophe'' shows that LZ'78 is not robust with respect to the
addition or deletion of bits. Since a usual good behaviour of functions used in
data representation is a kind of ``continuity'', our results show that, in this
respect, LZ'78 is not a good choice, as two words that differ in a single
bit can have images very far apart.

\subsection*{Organization of the paper}

In Section~\ref{sec:notations} we introduce all the notions related to
LZ'78 and state our main results (Section~\ref{sec:results}).
Section~\ref{sec:upperbound} is devoted to the proof of the upper
bound (the ``not a tragedy'' part), whereas the rest of the paper is
about lower bounds. In Section~\ref{sec:toy} we explicitly give a
word, based on de Bruijn sequences, whose compression ratio is optimal
but the addition of a single bit deteriorates the compression ratio as
much as the aforementioned upper bounds allows to. That is a
particular case of the result of Section~\ref{sec:general} but we
include it anyway for three reasons: it illustrates the main ideas
without obscuring them with too many technical details; the
construction is more explicit; and the bounds are better.

In Section~\ref{sec:general} we prove our main theorem on finite words
(Theorem~\ref{thm:lowerbound}). It requires the existence of a family of ``de
Bruijn-style'' words shown in Section~\ref{sec:family-de-bruijn} thanks to the
probabilistic method. Finally, Section~\ref{sec:infinite} uses the previous
results to prove the ``original'' one-bit catastrophe, namely on infinite
words (Theorem~\ref{thm:infinite}).

\section{Lempel-Ziv, compression and results}
\label{sec:notations}

Before turning to the description of LZ'78 algorithm, let us recall standard
notations on words.

\subsection{Basic notations}

The \emph{binary alphabet} is the set $\bin$. A \emph{word} $w$ is an
element of $\{0,1\}^\star$, that is, a finite ordered sequence of
letters $0$ or $1$, whose \emph{length} is denoted by $|w|$. The empty
word is denoted by $\ew$.  For a word $w = x_0\cdots x_{n-1}$ (note
that the indices begin at zero), where $x_i \in \{0,1\}$, $w[i..j]$
will denote the \emph{substring} $x_i\cdots x_j$ of $w$ (or $\ew$ if
$j<i$); $w[i]$ or $w_i$ will denote the letter $x_i$; and $w_{\leq i}$
(respectively $w_{< i}$) will denote $w[0..i]$ (resp. $w[0..i-1]$). We
say that a word $m$ is a \emph{factor} of $w$ if $m$ is any substring
$w[i..j]$. In the particular case of $i = 0$ (respectively $j = n-1$),
$m$ is also called a \emph{prefix} (resp. a \emph{suffix}) of
$w$. The set of factors of $w$ is denoted by $\f(w)$, and its set of
prefixes $\p(w)$. By extension, for a set $M$ of words, $\f(M)$ will denote
$\cup_{w \in M} \f(w)$ and similarly for $\p(M)$. If $u$ and $w$ are
two words, we denote by $\occ_w(u)$ the number of occurrences of the
factor $u$ in $w$.

The ``length-lexicographic order'' on words is the lexicographic order where
lengths are compared first.

An \emph{infinite word} is an element of $\{0,1\}^\nn$. The same notations as
for finite words apply.

All logarithms will be in base $2$. The size of a finite set $A$ is written
$|A|$.

\subsection{LZ'78}

\subsubsection{Notions relative to LZ}

A \emph{$k$-partition} (or just \emph{partition}) of a word $w$ is a sequence
of $k$ non-empty words $m_1,\dots, m_k$ such that
$w = m_1.m_2.\cdots.m_k$. The \emph{LZ-parsing} (or just \emph{parsing})
of a word $w$ is the unique partition of $w=m_1\cdots m_k$ such that:
\begin{itemize}
\item $m_1,\dots,m_{k-1}$ are all distinct\footnote{The last word $m_k$ might be
    equal to another $m_i$.};
\item $\forall i\leq k$, $\p(m_i)\subseteq \{m_1,\dots,m_i\}$.
\end{itemize}
The words $m_1,\dots, m_k$ are called \emph{blocks}. The \emph{predecessor} of
a block $m_i$ is the unique $m_j$, $j<i$, such that $m_i=m_ja$ for a letter
$a$. The compression algorithm LZ'78 parses the word $w$ and encodes each
block $m_i$ as a pointer to its predecessor $m_j$ together with the letter $a$
such that $m_i=m_ja$. For instance, the word $w=00010110100001$ is parsed as
$$\begin{array}{l|c|c|c|c|c|c|c}
  \text{Blocks} & 0 & 00 & 1 & 01 & 10 &100 & 001\\ \hline
  \text{Block number} & 0 & 1 & 2 & 3 & 4 & 5 & 6
\end{array}$$
and thus encoded as
$$(\ew,0);(0,0);(\ew,1);(0,1);(2,0);(4,0);(1,1).$$

The \emph{dictionary} of $w$ is the set
$\dic(w)=\{m_1,\dots, m_k\}$ (in the example, $\{0,1,00,01,10,001,100\}$).
Remark that, by definition, $\{\ew\}\cup\dic(w)$ is prefix-closed.

The \emph{parsing tree} of $w$ is the unique rooted binary tree $\tree(w)$
whose $(k+1)$ vertices are labeled with $\ew,m_1,\dots,m_k$, such that the
root is $\ew$ and if a vertex $m_i$ has a left child, then it is $m_i0$, and
if it has a right child, then it is $m_i1$.\footnote{Note that, in order to
  recover the parsing from the parsing tree, the vertices must also be labeled
  by the order of apparition of each block, but we do not need that in the
  sequel.} See Figure~\ref{fig:parsingtree}. 
Remark also that the depth of a vertex is equal to
the size of the corresponding block.

\begin{figure}[h]\centering
\tikzstyle{noeud}=[draw,circle,fill,draw=none,minimum size=5pt,inner sep=0pt]
\tikzstyle{noeudr}=[draw,circle,fill=green!70!blue!50,draw=none]
\tikzstyle{noeudg}=[draw,circle,fill=blue!100!green!60,draw=none]

\begin{tikzpicture}[scale=1, every node/.style={scale=1}]

  \node[noeud,label=above:$\lambda$] (Empty) at (0,0) {};

  \node[noeud,label=left:$0$] (R0) at (-1,-1) {};
  \node[noeud,label=right:$1$] (R1) at (1,-1) {};

  \node[noeud,label=left:$00$] (Q0) at (-1.4,-2) {};
  \node[noeud,label=right:$01$] (Q1) at (-0.6,-2) {};
  \node[noeud,label=right:$10$] (Q2) at (0.6,-2) {};

  \node[noeud,label=left:$001$] (Z0) at (-1,-3) {};
  \node[noeud,label=right:$100$] (Z1) at (0.2,-3) {};

  \draw[] (Empty) -- (R0);
  \draw[] (Empty) -- (R1);
  \draw[] (R0) -- (Q0);
  \draw[] (R0) -- (Q1);
  \draw[] (R1) -- (Q2);

  \draw[] (Q2) -- (Z1);
  \draw[] (Q0) -- (Z0);

  \end{tikzpicture}
  \caption{Parsing tree of $00010110100001$.}
  \label{fig:parsingtree}
\end{figure}
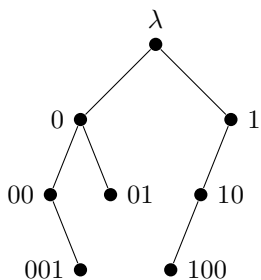

By abuse of language, we say that a block $b$ ``increases'' or ``grows'' in
the parsing of a word $w$ when we consider one of its successors, or when we
consider a path from the root to the leaves that goes through $b$. Indeed,
going from $b$ to its successor amounts to add a letter at the end of $b$
(hence the ``increase'').

\subsubsection{Compression ratio}

As in the example above, given a word $w$ and its LZ-parsing $m_1\cdots m_k$,
the \emph{LZ-compression} of $w$ is the ordered list of $k$ pairs $(p_i,a_i)$,
where $p_i$ is the binary representation of the unique integer $j<i$ such that
$m_j = m_i[0..(|m_i|-2)]$, and $a_i$ the last letter of $m_i$ (that is, the
unique letter such that $m_i = m_ja_i$). When the LZ-compression is given, one
can easily reconstruct the word $w$.
\begin{remark}\label{rem:extremal}
  \begin{itemize}
  \item If $x$ is a word, we define $\pref(x)$ the concatenation of all its
    prefixes in ascending order, that is,
    $$\pref(x)=x_0.x_0x_1.x_0x_1x_2.\cdots.x_0\cdots x_{n-2}x_{n-1}.$$
    Then the parsing of the word $w=\pref(x)$ is exactly the prefixes of $x$,
    thus the size of the blocks increases each time by one: this is the
    optimal compression. In that case, the number of blocks is
    $$k=|x|=\sqrt{2}\sqrt{|w|}-O(1).$$
    Actually, it is easy to see that this optimal compression is attained only
    for the words $w$ of the form $\pref(x)$.

    In Section~\ref{sec:general} we will need the concatenation of all
    prefixes of $x$ starting from a size $p+1$, denoted by $\pref_{>p}(x)$,
    that is,
    $$\pref_{>p}(x)=x_0x_1\cdots x_p.x_0x_1\cdots x_{p+1}.\cdots.x_0\cdots
    x_{n-1}.$$
  \item On the other hand, if $w$ is the concatenation, in length-lexicographic
    order, of all words of size $\leq n$ ($w=0.1.00.01.10.11.000.001\dots$),
    then it has size
    $$|w|=\sum_{i=1}^n i2^i=(n-1)2^{n+1}+2,$$
    and its parsing consists of all the words up to size $n$, therefore that
    is the worst possible case and the number of blocks is
    $$k=2^{n+1}-2=\frac{|w|}{\log|w|}+O\left(\frac{|w|}{\log^2|w|}\right).$$
    (And that is clearly not the only word achieving this worst compression.)
  \end{itemize}
\end{remark}
The number of bits needed in the LZ-compression is
$\Theta(\sum_{i=1}^k (|p_i| + 1))=\Theta(k\log k)$. As the two previous
extremal cases show, $k\log k=\Omega(\sqrt{|w|}\log|w|)$ and $k\log k=O(|w|)$.
\begin{definition}
  The \emph{compression ratio} of a word $w$ is
  $$\comp(w)=\frac{|\dic(w)|\log |\dic(w)|}{|w|}.$$
\end{definition}
As Remark~\ref{rem:extremal} shows,
$$\comp(w)=\Omega\left(\frac{\log|w|}{\sqrt{|w|}}\right)
\quad\text{and}\quad
\comp(w)\leq 1+O\left(\frac{1}{\log|w|}\right).$$

A sequence of words $(w_n)$ is said LZ-compressible if $\comp(w_n)$ tends to
zero (i.e. $k_n\log k_n=o(|w_n|)$), and consistently it will be considered
\emph{LZ-incompressible} if $\liminf_{n\to\infty}\comp(w_n)>0$ (in other
terms, $k_n \log k_n=\Omega(|w_n|)$).

Actually, the $(\log k)$ factor is not essential in the analysis of the
algorithm, therefore we drop it in our definitions (moreover, most of the time
we will focus directly on the size of the dictionary rather than the
compression ratio).
\begin{definition}
  The \emph{size of the LZ-compression} of $w$ (or \emph{compression size}, or
  also \emph{compression speed} when speaking of a sequence of words) is
  defined as the size of $\dic(w)$, that is, the number of blocks in the
  LZ-parsing of $w$.
\end{definition}
Remark that $|\dic(w)| = \Omega(\sqrt{|w|})$ and
$|\dic(w)| = O(|w|/\log(|w|))$. We can now restate the definition of
incompressibility of a sequence of words in terms of compression speed instead
of the number of bits in the LZ-compression.
\begin{definition}
  A sequence of words $(w_n)$ is said \emph{incompressible} iff
  $$|\dic(w_n)| = \Theta\left(\frac{|w_n|}{\log(|w_n|)}\right).$$
\end{definition}

In those definitions, we have to speak of sequences of finite words since the
asymptotic behaviour is considered. That is not needed anymore for infinite
words, of course, but then two notions of compression ratio are defined,
depending on whether we take the $\liminf$ or $\limsup$ of the compression
ratios of the prefixes.
\begin{definition}\label{def:infinite-compression-ratio}
  Let $w\in\{0,1\}^\nn$ be an infinite word.
  $$\compinf(w)=\liminf_{n\to\infty}\comp(w_{<n})\quad\text{and}\quad
  \compsup(w)=\limsup_{n\to\infty}\comp(w_{<n}).$$
\end{definition}

\subsection{One-bit catastrophe and results}
\label{sec:results}

The one-bit catastrophe question is originally stated only on infinite
words. It asks whether there exists an infinite word $w$ whose
compression ratio changes when a single letter is added in front of
it. More specifically, a stronger version asks whether there exists an
infinite word $w$ compressible (compression ratio equal to $0$) for
which $0w$ is not compressible (compression ratio $>0$). At
Section~\ref{sec:infinite} we will answer positively that question:
\begin{theorem}\label{thm:infinite}
  There exists $w\in\{0,1\}^\nn$ such that
  $$\compsup(w)=0\quad\text{and}\quad\compinf(w)\geq\frac{1}{6\,075}.$$
\end{theorem}
Remark that the $\liminf$ is considered for the compression ratio of $0w$ and
the $\limsup$ for $w$, which is the hardest possible combination as far as
asymptotic compression ratios are concerned.

But before proving this result, most of the work will be on finite words (only
in Section~\ref{sec:infinite} will we show how to turn to infinite words). Let
us therefore state the corresponding results on finite words. Actually, on
finite words we can have much more precise statements and therefore the
results are interesting on their own (perhaps even more so than the infinite
version).

In Section~\ref{sec:upperbound}, we show that the compression ratio of $aw$
cannot be much more than that of $w$. In particular, all words
``sufficiently'' compressible (compression speed $o(|w|/\log^2|w|)$) cannot
become incompressible when a letter is added in front (in some sense, thus,
the one-bit catastrophe cannot happen for those words, see
Remark~\ref{rk:catastrophe}).
\begin{theorem}
\label{thm:upperbound}
  For all word $w\in\{0,1\}^\star$ and any letter $a\in\{0,1\}$,
  $$|\dic(aw)|\leq 3\sqrt{|w|.|\dic(w)|}.$$
\end{theorem}
\begin{remark}
  When stated in terms of compression ratio, using the fact that
  $|\dic(w)|\geq \sqrt{|w|}$, this result reads as follows:
  $$\comp(aw)\leq 3\sqrt{2}\sqrt{\comp(w)\log|w|}.$$
\end{remark}

We also show in Section~\ref{sec:toy} that this result is tight up to a
multiplicative constant, since Theorem~\ref{thm:toy} implies the following
result.
\begin{theorem}
\label{thm:toy-lowerbound}
  For an infinite number of words $w\in\{0,1\}^\star$,
  $$|\dic(0w)|\geq \frac{1}{35}\sqrt{|w|.|\dic(w)|}.$$
\end{theorem}

More generally, we prove in Section~\ref{sec:general} our main result:
\begin{theorem}
  \label{thm:lowerbound}
  Let $l:\nn\to\nn$ be a function satisfying
  $l(n)\in[(90\log n)^2,\sqrt{n}]$. Then for an infinite number of words $w$:
  $$|\dic(w)|\leq\frac{3+\sqrt{3}}{2}\cdot\frac{|w|}{l(|w|)}\;\text{ and }\;
  |\dic(0w)|\geq\frac{1}{54}\cdot\frac{|w|}{\sqrt{l(|w|)}}.$$
\end{theorem}
This shows that the upper bound is tight (up to a multiplicative
constant) for any possible compression speed. This also provides an
example of compressible words that become incompressible when a letter
is added in front (see Remark~\ref{rk:catastrophe}), thus showing the
one-bit catastrophe for finite words.

\begin{remark}\label{rk:catastrophe}
  In particular:
  \begin{itemize}
  \item Theorem~\ref{thm:upperbound} implies that, if an increasing sequence
    of words $(w_n)$ satisfies $|\dic(w_n)|=o(|w_n|/\log^2|w_n|)$, then for
    any letter $a\in\{0,1\}$, $aw_n$ remains fully compressible
    ($|\dic(aw_n)|=o(|w_n|/\log|w_n|)$);
  \item however, by Theorem~\ref{thm:lowerbound}, there is an increasing
    sequence of words $(w_n)$ such that
    $|\dic(w_n)|=\Theta(|w_n|/\log^2|w_n|)$ (compressible) but
    $|\dic(0w_n)|=\Theta(|w_n|/\log |w_n|)$ (incompressible), which is the
    one-bit catastrophe on finite words;
  \item the following interesting case is also true: there is an increasing
    sequence of words $(w_n)$ such that $|\dic(w_n)|=\Theta(\sqrt{|w_n|})$
    (optimal compression) but $|\dic(0w_n)|=\Theta(|w_n|^{3/4})$. This special
    case is treated extensively in Theorem~\ref{thm:toy}.
  \end{itemize}
\end{remark}

\subsection{Parsings of $w$ and $aw$}
We will often compare the parsing of a word $w$ and the parsing of $aw$
for some letter~$a$: let us introduce some notations (see Figure~\ref{fig:blocks}).
\begin{itemize}
\item The blocks of $w$ will be called the \emph{green blocks}.
\item The blocks of $aw$ will be called the \emph{red blocks} and are split
  into two categories\footnote{Except the first block of $aw$, which is the
    word $a$ and which is just called a red block.}:
  \begin{itemize}
  \item The \emph{junction blocks}, which are red blocks that overlap two or
    more green blocks when we align $w$ and $aw$ on the right (that is, the
    factor $w$ of $aw$ is aligned with the word $w$, see
    Figure~\ref{fig:blocks}).
  \item The \emph{offset-$i$ blocks}, starting at position $i$ in a
    green block and completely included in it. If not
    needed, the parameter $i$ will be omitted.
  \end{itemize}
\end{itemize}

\begin{figure}[h]\centering
\begin{tikzpicture}[scale=1, every node/.style={scale=1}]
  \draw [fill=green!70!blue!50,rounded corners=0.5mm] (-0.17,-0.17) rectangle (0.17,0.17);
  \draw [fill=green!70!blue!50,rounded corners=0.5mm] (-0.17+0.4,-0.17) rectangle (2*0.4+0.17,0.17);
  \draw [fill=green!70!blue!50,rounded corners=0.5mm] (-0.17+3*0.4,-0.17) rectangle (5*0.4+0.17,0.17);
  \draw [fill=green!70!blue!50,rounded corners=0.5mm] (-0.17+6*0.4,-0.17) rectangle (6*0.4+0.17,0.17);
  \draw [fill=green!70!blue!50,rounded corners=0.5mm] (-0.17+7*0.4,-0.17) rectangle (8*0.4+0.17,0.17);
  \draw [fill=green!70!blue!50,rounded corners=0.5mm] (-0.17+9*0.4,-0.17) rectangle (11*0.4+0.17,0.17);

  \draw [fill=red!50,rounded corners=0.5mm] (-0.4-0.17,-1-0.17) rectangle (-0.4+0.17,-1+0.17);
  \draw [fill=red!50,rounded corners=0.5mm] (-0.17,-1-0.17) rectangle (0.4+0.17,-1+0.17);
  \draw [fill=red!50,rounded corners=0.5mm] (2*0.4-0.17,-1-0.17) rectangle (2*0.4+0.17,-1+0.17);
  \draw [fill=red!50,rounded corners=0.5mm] (3*0.4-0.17,-1-0.17) rectangle (4*0.4+0.17,-1+0.17);
  \draw [fill=red!50,rounded corners=0.5mm] (5*0.4-0.17,-1-0.17) rectangle (7*0.4+0.17,-1+0.17);
  \draw [fill=red!50,rounded corners=0.5mm] (8*0.4-0.17,-1-0.17) rectangle (10*0.4+0.17,-1+0.17);
  \draw [fill=red!50,rounded corners=0.5mm] (11*0.4-0.17,-1-0.17) rectangle (11*0.4+0.17,-1+0.17);

  \draw[dotted] (2*0.4-0.17,-1+0.18) -- (2*0.4-0.17,-0.18);
  \draw[dotted] (2*0.4+0.17,-1+0.18) -- (2*0.4+0.17,-0.18);
  \fill[gray!20,opacity=0.60] (2*0.4-0.17,-1+0.18) rectangle  (2*0.4+0.17,0.15);

  \draw[dotted] (3*0.4-0.17,-1+0.18) -- (3*0.4-0.17,-0.18);
  \draw[dotted] (4*0.4+0.17,-1+0.18) -- (4*0.4+0.17,-0.18);
  \fill[gray!20,opacity=0.60] (3*0.4-0.17,-1+0.18) rectangle  (4*0.4+0.17,0.15);

  \draw[dotted] (8*0.4-0.17,-1+0.18) -- (8*0.4-0.17,-0.18);
  \draw[dotted] (10*0.4+0.17,-1+0.18) -- (10*0.4+0.17,-0.18);
  \fill[gray!20,opacity=0.60] (8*0.4-0.17,-1+0.18) rectangle  (10*0.4+0.17,0.15);

  \node[align=center] (A) at (0,0) {0};
  \node[align=center] (A) at (0.4,0) {0};
  \node[align=center] (A) at (0.8,0) {1};
  \node[align=center] (A) at (1.2,0) {0};
  \node[align=center] (A) at (1.6,0) {1};
  \node[align=center] (A) at (2,0) {0};
  \node[align=center] (A) at (2.4,0) {1};
  \node[align=center] (A) at (2.8,0) {0};
  \node[align=center] (A) at (3.2,0) {0};
  \node[align=center] (A) at (3.6,0) {0};
  \node[align=center] (A) at (4,0) {1};
  \node[align=center] (A) at (4.4,0) {1};

  \node[align=center] (A) at (-0.4,-1) {0};
  \node[align=center] (A) at (0,-1) {0};
  \node[align=center] (A) at (0.4,-1) {0};
  \node[align=center] (A) at (0.8,-1) {1};
  \node[align=center] (A) at (1.2,-1) {0};
  \node[align=center] (A) at (1.6,-1) {1};
  \node[align=center] (A) at (2,-1) {0};
  \node[align=center] (A) at (2.4,-1) {1};
  \node[align=center] (A) at (2.8,-1) {0};
  \node[align=center] (A) at (3.2,-1) {0};
  \node[align=center] (A) at (3.6,-1) {0};
  \node[align=center] (A) at (4,-1) {1};
  \node[align=center] (A) at (4.4,-1) {1};

  \node[align=center] (ST) at (-0.3,-2) {offset-$1$};

  \node[align=center] (ST2) at (1.4,-2) {offset-$0$};

  \node[align=center] (ST3) at (3.6,-2) {junction};

  \draw[->,>=latex] (ST) to[bend right] (0.8,-1.2);
  \draw[->,>=latex] (ST2) to (1.4,-1.2);
  \draw[->,>=latex] (ST3) to (3.6,-1.2);

\end{tikzpicture}
\caption{The green blocks of $w$ and red blocks of $0w$ for
  $w = 001010100011$.}
\label{fig:blocks}
\end{figure}
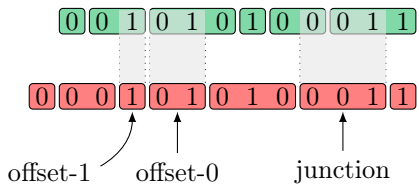

\section{Upper bound}
\label{sec:upperbound}

This section is devoted to the proof of Theorem~\ref{thm:upperbound} giving an
upper bound on the compression ratio of $aw$, for any letter $a$, as a
function of the compression ratio of the word $w$. In their 1998
paper~\cite{LathropStrauss}, Lathrop and Strauss ask the following question: ``Consider
optimally compressed sequences: Will such sequences compress reasonably well
if a single bit is removed or added to the front of the sequence?'' We give a
positive and quantified answer: indeed, a word $w$ compressed optimally has a
compression speed $O(\sqrt{n})$, thus by Theorem~\ref{thm:upperbound}, the
word $aw$ has a compression speed $O(n^{3/4})$. (And we shall complete this
answer with the matching lower bound in the next section.)

The first lemma bounds the size of the partition of a word $w$ if the
partitioning words come from a family with a limited number of words of same
size. In its application, the partition will be a subset of the LZ-parsing, and
Lemma~\ref{lem:factortree} below will give the required bound on the number of
factors of a given size.
\begin{lemma}
\label{lem:worstcase}
Let $\F$ be a family of distinct words such that for each $i$, the
number of words of size $i$ in $\F$ is bounded by a constant
$N$. Suppose that a word $w$ is partitioned into different words of
$\F$. Then the number of words used in the partition is at most
$2\sqrt{N|w|}$.
\end{lemma}
\begin{proof}
  Let $m(i)$ be the number of words of size $i$ occurring in the partition of
  $w$, and $k$ the size of the largest words used. We want to prove that
  $$\sum_{i=1}^k m(i)\leq 2\sqrt{N|w|}.$$
  We have:
  $$|w|=\sum_{i=1}^k im(i) \geq \sum_{i\geq\sqrt{\frac{|w|}{N}}} im(i)\geq
  \sqrt{\frac{|w|}{N}}\sum_{i\geq\sqrt{\frac{|w|}{N}}} m(i)$$
  hence
  $$\sum_{i\geq\sqrt{\frac{|w|}{N}}} m(i)\leq \sqrt{N|w|}.$$
  On the other hand, since $m(i)\leq N$:
  $$\sum_{i<\sqrt{\frac{|w|}{N}}} m(i) < N\sqrt{\frac{|w|}{N}}=\sqrt{N|w|}.$$
\end{proof}
\begin{remark}
  Note that if, for all $i\geq 1$, $\F$ contains exactly $\min(2^i,N)$ words
  of size $i$, the concatenation of all the words of $\F$ up to size $s$ gives
  a word $w$ of size
  $$|w|=\sum_{i=1}^{\log N}i2^i+\sum_{i>\log N}^s iN\leq 2N\log N+(s-\log N)(s+\log N + 1)N/2$$
  partitioned into $m$ blocks, where
  $$m=\sum_{i=1}^{\log N}2^i+\sum_{i>\log N}^s N\geq (s-\log N)N.$$
  Thus $m\geq \sqrt{2}\sqrt{N|w|}$ if $s>>\log N$. This shows the optimality
  of Lemma~\ref{lem:worstcase} up to a factor $\sqrt{2}$.
\end{remark}

We now come to the lemma bounding the number of factors of a given size in a
word $w$ as a function of its LZ-parsing.
\begin{lemma}
\label{lem:factortree}
 Let $T$ be the parsing tree of a word $w$. Then the number of different
  factors of size $i$ in the blocks of $w$ is at most $|T|-i$ (that is,
  $|\f(\dic(w))\cap\{0,1\}^i|\leq |T|-i$).
\end{lemma}

\begin{proof}
  A factor of size $i$ in a block $b$ corresponds to a subpath of size $i$ in
  the path from the root to $b$ in the parsing tree. The number of such
  subpaths is bounded by the number of vertices at depth at least $i$.
\end{proof}
Actually, below we will use Lemma~\ref{lem:factortree} sub-optimally since we
will ignore the parameter $i$ and use the looser bound $(|T|-1)$.

Let us turn to the proof of Theorem~\ref{thm:upperbound}, the main result
of the present section.
\begin{proof}[Proof of Theorem~\ref{thm:upperbound}]
  Let $D = \dic(aw)$ be the set of red blocks. We partition $D$ into
  $D_1$ and $D_2$, where $D_1$ is the set of junction blocks together
  with the first red block (consisting only of the letter $a$), and
  $D_2$ is the set of offset blocks.

\begin{itemize}
\item Bound for $D_1$: The number of junction blocks is less than the
  number of green blocks, therefore
  $|D_1| \leq |\dic(w)|\leq \sqrt{|\dic(w)|.|w|}$ (recall that $|\dic(w)|\leq|w|$).
\item Bound for $D_2$: Consider $\tilde{w}$ the word $w$ where all the
  junction blocks have been replaced by the empty word $\lambda$. We know that
  $\tilde{w}$ is partitioned into different words by $D_2$. But
  $D_2 \subset \F$, where $\F = \f(\dic(w))$ (the set of factors contained in
  the green blocks). By Lemma \ref{lem:factortree}, the number of words of
  size $i$ in $\F$ is bounded by $|\tree(w)|-i$, which is at most $|\dic(w)|$.
  Finally, Lemma \ref{lem:worstcase} tells us that the number of words in any
  partition of $\tilde{w}$ by words of $\F$ is bounded by
  $2\sqrt{|\dic(w)|.|\tilde{w}|} \leq 2\sqrt{|\dic(w)|.|w|}$.
\end{itemize}
In the end, $|D| = |D_1| + |D_2| \leq 3\sqrt{|w|.|\dic(w)|}$.
\end{proof}

\begin{remark}
  Instead of a single letter, we can add a whole word $z$ in front of~$w$.
  With the same proof, it is easy to see that
  $$|\dic(zw)|\leq |\dic(z)|+3\sqrt{|w|.|\dic(w)|}.$$
  
  Alternately, if we remove the first letter of $w=aw'$ (or any prefix) we get
  the same upper bound:
  $$|\dic(w')|\leq 3\sqrt{|aw'|.|\dic(aw')|}.$$
\end{remark}

\section{``Weak catastrophe'' for the optimal compression ratio}
\label{sec:toy}

Before the proof of Theorem \ref{thm:lowerbound}, we first present a
``weak catastrophe'', namely the third item of Remark~\ref{rk:catastrophe}
in which the compression speed of a sequence changes from $O(\sqrt{n})$
(optimal compression) to $\Omega(n^{3/4})$ when a letter is added in front,
thus matching the upper bound of Theorem~\ref{thm:upperbound}.
\begin{theorem}\label{thm:toy}
  For an infinite number of words $w$:
  $$|\dic(w)|\leq 1.9\sqrt{|w|}\;\text{ and }\;
  |\dic(0w)|\geq 0.039|w|^{3/4}.$$ 
\end{theorem}
\begin{remark}
  The ``true'' values of the constants that we will get below are as follows:
  $$|\dic(w)|\leq 3\sqrt{\frac{2}{5}}\sqrt{|w|}\;\text{ and }\;
  |\dic(0w)|\geq \frac{1}{36}\left(\frac{8}{5}\right)^{3/4}|w|^{3/4}-o(|w|^{3/4}).$$
\end{remark}

Observe that this weak catastrophe is a special case of Theorem
\ref{thm:lowerbound} (with better constants, though). The aim of this section
is twofold: first, it will be a constructive proof, whereas the main theorem
will use the probabilistic method; second, this section will set up the main
ideas and should help understand the general proof.

A main ingredient in the construction is de Bruijn sequences, that we
introduce shortly before giving the overview of the proof.

\subsection{De Bruijn sequences}
A \emph{de Bruijn sequence} of order $k$ (or $\DB(k)$ in short,
notation that will also designate the set of all de Bruijn sequences
of order $k$) is a word $x$ of size $2^k + k-1$ in which every word of
size $k$ occurs exactly once as a substring. For instance,
\texttt{0001011100} is an example of a $\DB(3)$. Such words exist for
any order $k$ as they are, for instance, Eulerian paths in the regular
directed graph whose vertices are words of size $(k-1)$ and where
there is an arc labeled with letter $a$ from $u$ to $v$ iff
$v=u[1..k-2]a$.

Given any $x\in\DB(k)$, the following well-known (and straightforward) property
holds:
\begin{center}
  \pstar Any word $u$ of size at most $k$ occurs exactly $2^{k-|u|}$
  times in $x$.
\end{center}
(In symbols, $\occ_x(u)=2^{k-|u|}$.) Thus, a factor of size $l\leq k$ in $x$
will identify exactly $2^{k-l}$ positions in $x$ (the $i$-th position is the
beginning of the $i$-th occurence of the word).

The use of de Bruijn sequences is something common in the study of this kind
of algorithms: Lempel and Ziv themselves use it in~\cite{LZ76}, as well as
later~\cite{LathropStrauss} and~\cite{Pierce2000} for example.

\subsection{Overview of the proof}

Recall that a word $w$ is optimally compressed iff it is of the form
$w=\pref(x)$ for some word $x$ (Remark~\ref{rem:extremal}). Thus we are
looking for an $x$ such that $0\pref(x)$ has the worst possible compression
ratio. In Section~\ref{sec:upperbound} the upper bound on the dictionary size
came from the limitation on the number of possible factors of a given size: it
is therefore natural to consider words $x$ where the number of factors is
maximal, that is, de Bruijn sequences.

Although we conjecture that the result should hold for $w=\pref(x)$ whenever
$x$ is a de Bruijn sequence beginning with $0$, we were not able to show it
directly. Instead, we need to (possibly) add small words, that we will call
``gadgets'', between the prefixes of $x$.

For some arbitrary $k$, we fix $x\in\DB(k)$ and start with the word
$w=\pref(x)$ of size $n$. The goal is to show that there are
$\Omega(n^{3/4})$ red blocks (i.e that the size of the dictionnary for
$0w$ is $\Omega(n^{3/4})$): this will be achieved by showing that a
significant (constant) portion of the word $0w$ is covered by
``small'' red blocks (of size $O(n^{1/4})$). Let $s=|x|$, so that
$n=\Theta(s^2)$. More precisely, we show that, in all the prefixes $y$
of $x$ of size $\geq 2s/3$, at least the last third of $y$ is covered
by red blocks of size $O(\sqrt{s})=O(n^{1/4})$.

This is done by distinguishing between red blocks starting near the beginning of a
green block (offset-$i$ for $i\leq\gamma k$) and red blocks starting at
position $i>\gamma k$:
\begin{itemize}
\item For the first, what could happen is that by coincidence the
  parsing creates most of the time an offset-$i$ red block (called
  $i$-violation in the sequel), which therefore would increase until
  it covers almost all the word $w$. To avoid this, we introduce
  gadgets: we make sure that this happens at most half of the time
  (and thus cannot cover more than half of $w$). More precisely,
  Lemma~\ref{lemma:fewviolations} shows that at most half of the
  prefixes of $x$ can contain offset-$i$ blocks for any fixed
  $i\leq\gamma k$. This is due to the insertion of gadgets that
  ``kill'' some starting positions $i$ if necessary, by
  ``resynchronizing'' the parsing at a different position.
\item On the other hand, red blocks starting at position $i>\gamma k$ are
  shown to be of small size by Proposition~\ref{prop:smallblocks}. This is
  implied by Lemma~\ref{lem:junctions} claiming that, due to the structure of
  the $\DB(k)$ (few repetitions of factors), few junction red blocks can go up
  to position $(i-1)$ and precede an offset-$i$ block.
\end{itemize}
Since all large enough prefixes of $x$ have a constant portion containing only
red blocks of size $O(n^{1/4})$, the compression speed is $\Omega(n^{3/4})$
(Theorem~\ref{thm:toy}).

Gadgets must satisfy two conditions:
\begin{itemize}
\item they must not disturb the parsing of $w$;
\item the gadget $g_i$ must ``absorb'' the end of the red block ending
  at position $(i-1)$, and ensures that the parsing restarts at a
  controlled position different from $i$.
\end{itemize}
The insertion of gadgets in $w$ is not trivial because we need to ``kill''
positions without creating too many other bad positions, that is why gadgets
are only inserted in the second half of $w$. Moreover, gadget insertion
depends on the parsing of $0w$ and must therefore be adaptative, which is the
reason why we give an algorithm to describe the word $w$.

Let us summarize the organisation of the lemmas of this section:
\begin{itemize}
\item Lemma~\ref{lemma:violations} is necessary for the algorithm: it shows
  that, in $0\pref(x)$, there can be at most one position $i$ such that the
  number of $i$-violations is too high.
\item Lemma~\ref{lem:toy-upper-bound} shows that the parsing of $w$ is not
  disturbed by gadgets and therefore the compression speed of $w$ is
  $O(\sqrt{n})$.
\item Lemma~\ref{lemma:fewviolations} shows that gadgets indeed remove
  $i$-violations as required, for $i\leq\gamma k$.
\item Lemma~\ref{lem:junctions} uses the property of the $\DB(k)$ to prove
  that junction blocks cannot create too many $i$-violations if $i>\gamma k$.
\item Finally, Proposition~\ref{prop:smallblocks} uses
  Lemma~\ref{lem:junctions} to show that the offset-$i$ red blocks are small
  if $i$ is large.
\end{itemize}

\subsection{Construction and first properties}

Let $\gamma$ be any constant greater than or equal to $3$. Let $x$ be a
$\DB(k)$ beginning by $01$. We denote its size by $s=2^k+k-1$. Suppose for
convenience that $k$ is odd, so that $s$ is even.\footnote{This is to avoid
  dealing with the fractional part of $s/2$, but the construction also works
  in the case where $k$ is even.} For $i\in[0,s-1]$, let $w_i=x_{\leq i}$, so
that $\pref(x)=w_0.w_1\dots w_{s-1}$.

The word $w$ that we will construct is best described by an algorithm. It will
merely be $\pref(x)$ in which we possibly add ``gadgets'' (words) between some
of the $w_j$ in order to control the parsings of $w$ and $0w$. The letter in
front that will provoke the ``catastrophe'' is the first letter of $w$, that
is, $0$.

The gadgets $g_i^j$ (for $i\in[0,\gamma k]$ and $j\geq 0$) are defined as follows
(where $\bar x_i$ denotes the complement of $x_i$):
\begin{itemize}
\item $g_0^j=10^j$;
\item and for $i>0$, $g_i^j=x_{<i}.\bar x_i.1^j$.
\end{itemize}

Recall that the green blocks are those of the parsing of $w$, whereas the red
ones are those of the parsing of $0w$. We call ``regular'' the green blocks
that are not gadgets (they are of the form $w_j$ for some $j$). For
$i\in[0,s-1]$, we say that a regular green block in $w$ is $i$-violated
if there is an offset-$i$ (red) block in it. Note that gadgets do not count in
the definition of a violation.

\begin{lemma}\label{lemma:violations}
  For $i\in[0,s-1]$, let $l_i$ be the number of $i$-violated blocks in
  $\pref(x)=w_0.w_1\dots w_{s-1}$. Then for all $i\neq i'$,
  $l_i+l_{i'}\leq s$.

  In particular, there can be at most one $i$ such that the number of
  $i$-violated blocks is $>s/2$.
\end{lemma}

\begin{proof}
  Let $i$ and $i'$ be such that $0\leq i<i'<s$.

  Consider the red blocks starting at position $i$ and $i'$ in any green
  block.

  No green block in $w_0\dots w_{i'-1}$ is $i'$-violated since they are too
  small to contain position $i'$. Let $a$ be the number of $i$-violated blocks
  in $w_0\dots w_{i'-1}$. In $w_{i'}\dots w_{s-1}$, let $b$ be the number of
  green blocks that are both $i$-violated and $i'$-violated, and let $c$
  (respectively $d$) be the number of $i$-violated (resp. $i'$-violated)
  blocks that are not $i'$-violated (resp. $i$-violated) blocks.

  The number of $i$-violations is $l_i=a+b+c$ and the number of
  $i'$-violations is $l_{i'}=b+d$. But $b+c+d \leq s -i'$ and $b\leq i'-i-a$
  (since a red block starting at position $i$ can only be increased $(i'-i)$
  times before it overlaps position $i'$, and it has already increased $a$
  times in the first $i'$ green blocks), so that
  $l_i + l_{i'} = (b + c + d) + (a + b) \leq (s-i') + (i'-i) \leq s$.
  Therefore, $l_i$ or $l_{i'}$ has to be $\leq s/2$.
\end{proof}

The algorithm constructing $w$, illustrated in Figure~\ref{fig:algo}, is as
follows.
\begin{enumerate}
\item If the number of $i$-violations in $w_0.w_1\dots w_{s-1}$ is
  $\leq s/2$ for all $i\in[0,\gamma k]$, then output
  $w=w_0.w_1\dots w_{s-1}$.
\item Otherwise, let $i$ be the (unique by Lemma~\ref{lemma:violations}) integer
  in $[0,\gamma k]$ for which the number of $i$-violations is $>s/2$. Let
  $c=0$ (counter for the number of inserted gadgets) and $d=s/2+1$ (counter
  for the place of the gadget to be inserted).
\item For all $j\in[0,s-1]$, let $z_j=w_j$.
\item While the number of $i$-violations in $z_0.z_1\dots z_{s-1}$ is $\geq d$, do:
  \begin{enumerate}
  \item let $j$ be such that $w_j$ is the $d$-th $i$-violated green block;
  \item $z_j\leftarrow g_i^{c}w_j$ (we add the gadget $g_i^c$ before the block
    $w_j$);
  \item $c\leftarrow c+1$;
  \item if $w_j$ is still $i$-violated, then $d\leftarrow d+1$.
  \end{enumerate}
\item Return $w=z_0.z_1\dots z_{s-1}$.
\end{enumerate}

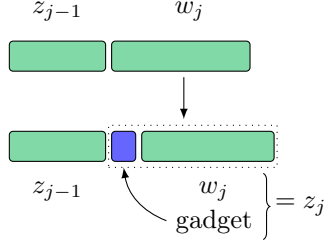
\begin{figure}[h]
  \centering
  \begin{tikzpicture}[scale=0.8, every node/.style={scale=1}]
    \node[] (W1) at (0.8,1) {$z_{j-1}$};
    \node[] (W2) at (3,1) {$w_{j}$};

  \draw [fill=green!70!blue!50,rounded corners=0.5mm] (0,0) rectangle (1.6,0.5);
  \draw [fill=green!70!blue!50,rounded corners=0.5mm] (1.7,0) rectangle (4,0.5);

  \draw [fill=green!70!blue!50,rounded corners=0.5mm] (0.0,-1.5) rectangle (1.6,-1);
  \draw [fill=blue!100!green!60,rounded corners=0.5mm] (1.7,-1.5) rectangle (2.1,-1);
  \draw [fill=green!70!blue!50,rounded corners=0.5mm] (2.2,-1.5) rectangle (4.4,-1);

  \draw [dotted] (1.65,-1.6) rectangle (4.45,-0.9);

    \draw[->,>=latex] (2.9,-0.1) -- (2.9,-0.8);

    \node[] (W1') at (0.8,-2) {$z_{j-1}$};
    \node[] (W2') at (3.4,-2) {$w_{j}$};
    \node[] (G) at (3.4,-2.5) {gadget};
    \draw[->,>=latex] (G) to [bend left] (1.9,-1.6);
    \draw[decoration={brace,raise=5pt},decorate]  (4,-1.7) --
    node[right=6pt] {$ = z_j$} (4,-2.8);
  \end{tikzpicture}
  \caption{Illustration of Step 4(b) of the algorithm.}
  \label{fig:algo}
\end{figure}

\begin{figure}[h]\centering
\tikzstyle{noeud}=[draw,circle,fill,draw=none]
\tikzstyle{noeudr}=[draw,circle,fill=green!70!blue!50,draw=none]
\tikzstyle{noeudg}=[draw,circle,fill=blue!100!green!60,draw=none]

\begin{tikzpicture}[scale=0.8, every node/.style={scale=0.4}]
    \draw [fill=green!70!blue!50,rounded corners=0.5mm] (-20,-4) rectangle (-19.5,-3.5);
  \draw [fill=green!70!blue!50,rounded corners=0.5mm] (-19.4,-4) rectangle (-18.8,-3.5);
  \draw [fill=green!70!blue!50,rounded corners=0.5mm] (-17.5,-4) rectangle (-16.5,-3.5);

  \draw [dotted] (-18.7, -3.8) -- (-17.6,-3.8);

  \draw [fill=blue!100!green!60,rounded corners=0.5mm] (-16.4,-4) rectangle (-16.2,-3.5);

  \draw [fill=green!70!blue!50,rounded corners=0.5mm] (-16.1,-4) rectangle (-15,-3.5);

  \draw [fill=blue!100!green!60,rounded corners=0.5mm] (-14.9,-4) rectangle (-14.5,-3.5);

  \draw [fill=green!70!blue!50,rounded corners=0.5mm] (-14.4,-4) rectangle (-13.2,-3.5);
  \draw [fill=green!70!blue!50,rounded corners=0.5mm] (-13.1,-4) rectangle (-11.8,-3.5);

  \draw [fill=blue!100!green!60,rounded corners=0.5mm] (-11.7,-4) rectangle (-11.2,-3.5);

  \draw [dotted] (-11.1, -3.8) -- (-9.9,-3.8);

  \draw [fill=green!70!blue!50,rounded corners=0.5mm] (-9.8,-4) rectangle (-7,-3.5);

  \draw [<-,>=latex] (-19.7,-3.4) to [bend left] (-19.1,-3.4);

  \draw [<-,>=latex] (-17,-3.4) to [bend left] (-15.7,-3.4);

  \draw [<-,>=latex] (-15.3,-3.4) to [bend left] (-13.9,-3.4);

  \draw [<-,>=latex] (-13.7,-3.4) to [bend left] (-12.5,-3.4);

  \draw [<-,>=latex] (-16.2,-4.1) to [bend right] (-14.8,-4.1);
  \draw [<-,>=latex] (-14.6,-4.1) to [bend right] (-11.4,-4.1);

 \begin{scope}[shift={(-5,-1)},scale=0.6]

  \node[noeud] (Empty) at (0,0) {};

  \node[noeudr] (R0) at (-1,-1) {};
  \node[noeudr] (R1) at (-1,-2) {};
  \node[noeudr] (R2) at (-1,-4) {};
  \node[noeudr] (R2') at (-1,-3) {};
  \node[noeudr] (R4) at (-1,-7) {};
  \node[noeudr] (R5) at (-1,-8) {};

  \node[noeudg] (G1) at (0,-3) {};
  \node[noeudg] (G2) at (0,-4) {};
  \node[noeudg] (G3) at (0,-6) {};

  \draw[] (Empty) -- (R0);
  \draw[] (R0) -- (R1);
  \draw[] (R1) -- (R2');
  \draw[] (R2') -- (R2);
  \draw[dotted] (R2) -- (R2');
  \draw[dotted] (R2) -- (R4);
  \draw[] (R4) -- (R5);

  \draw[] (R1) -- (G1);
  \draw[] (G1) -- (G2);
  \draw[dotted] (G2) -- (G3);
  \end{scope}

  \end{tikzpicture}
  \caption{Left: Form of the word $w$. The blocks in green are
    the regular blocks, the blocks in blue are the gadgets. The arcs
    represent the relation of paternity. Right: The shape of the
    parsing tree of $w$.}
  \label{fig:algo2}
\end{figure}
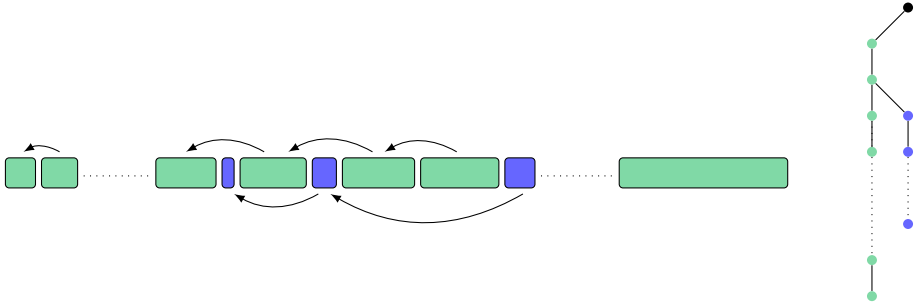

Some parts of the algorithm might seem obscure, in particular the role
of the counter $d$. The proof of the following properties should help
understand this construction, but let us first explain the intuition
behind the algorithm.  Below (Proposition~\ref{prop:smallblocks}) we
will have a generic argument (i.e. true without gadgets) to deal with
the $i$-violations for $i>\gamma k$, therefore for now we only care of
$i$-violations for $i\leq \gamma k$. They are not problematic if there
are at most (roughly) $s/2$ of them. Thanks to
Lemma~\ref{lemma:violations}, there is therefore at most one $i_0$
which can be problematic. To guarantee the upper bound of (roughly)
$s/2$ for the number of $i_0$-violations, every time it is necessary
we insert between two regular green blocks one gadget to kill the
$(s/2+1)$-th, $(s/2+2)$-th, etc., $i_0$-violations. But gadgets are
guaranteed to work as expected only if at least $1+(\gamma+1)k$ of
them have already been inserted (see Lemma~\ref{lemma:fewviolations}),
hence the counter $d$ is useful to avoid inserting two gadgets in
front of the same regular block.

From now on, we call $w$ the word output by the algorithm. We first evaluate
the size of $w$. Its minimal size is obtained when no gadgets are added during
the algorithm:
$$|w|\geq \frac{s(s+1)}{2}.$$
On the other hand, if $s/2$ gadgets $g_{\gamma k}^c$ of size $\gamma k+1+c$
are added, we obtain an upper bound on $|w|$:
$$|w|\leq \frac{s(s+1)}{2}+\sum_{c=0}^{s/2-1}(\gamma k+1+c)=\frac{5s^2}{8}+o(s^2).$$

Let us show that the word $w$ is nearly optimally compressible (upper
bound).
\begin{lemma}\label{lem:toy-upper-bound}
  The compression speed of $w$ is at
  most $$3\sqrt{\frac{2}{5}}\sqrt{|w|}.$$
\end{lemma}
\begin{proof}
  If the algorithm stops at step 1, then $w=\pref(x)$ and it is compressed
  optimally (see Remark~\ref{rem:extremal}): the compression speed is
  $$\sqrt{2}\sqrt{|w|}+O(1).$$

  Otherwise, we add at most one gadget for each $w_j$, and only for $j>s/2$.
  Therefore, there are at most $s/2$ gadgets. Remark that, for the $i$ fixed
  in the algorithm, the gadgets $(g_i^j)_j$ are prefixes one of each other,
  and none of them are prefixes of $x$. Thus the parsing tree of $w$ consists
  of one main path of size $s$ (corresponding to $w_0,w_1,\dots,w_{s-1}$),
  together with another path of size $\leq s/2$ (corresponding to the gadgets
  $(g_i^j)_j$) starting from a vertex of the main path. See Figure~\ref{fig:algo2}.

  The worst case for the compression speed is when the second path is of size
  $s/2$ and starts at the root. Then the size of $w$ is
  $$|w|\geq  \frac{s(s+1)}{2}+\frac{(s/2)(1+s/2)}{2}\geq \frac{5s^2}{8}$$
  and the size of the dictionary is $3s/2$, yielding the compression speed
  stated in the lemma.
\end{proof}

Let us now turn to the lower bound on the compression speed of $0w$. The next
lemma shows that, for $i\leq\gamma k$, there are not too many $i$-violations
thanks to the gadgets.

\begin{lemma}\label{lemma:fewviolations}
  For all $i\in[0,\gamma k]$, the number of $i$-violations in $w$ is at most
  $$s/2+(1+\gamma)k+1.$$
\end{lemma}
\begin{proof}
  If no gadgets have been added during the algorithm, then for all
  $i\in[0,\gamma k]$, the number of $i$-violations in $w$ is $\leq s/2$.

  Otherwise, first remark that Lemma~\ref{lemma:violations} remains valid even
  when the gadgets are added. We need to distinguish on the type ($i=0$ or
  $i>0$) of the most frequent violations in $w$.
  \begin{itemize}
  \item Case 1: the most frequent violations are $0$-violations. In that case,
    we claim that whenever a gadget is inserted before a block $w_i$, the
    $0$-violation in $w_i$ disappears. It is enough to prove that whenever a
    gadget $\ga{0}{j}$ is added, it was already in the dictionary of $0w$, so
    that the next word in the dictionary will begin by $\ga{0}{j}0$ and the
    parsing will overlap position $0$ of the next green block.

    We proceed by induction: for $j=0$, $\ga{0}{0} = 1$, and this word is the
    third block in the parsing of $0w$, because $x$ starts with $01$. For
    $j>0$: when $\ga{0}{j-1}$ was parsed, by induction it was already in the
    dictionary, so that the block added in the dictionary of $0w$ starts with
    $\ga{0}{j-1}0=\ga{0}{j}$.

    After at most $s/2$ iterations of the while loop, there is no more
    $(s/2+1)$-th $0$-violation: the number of $0$-violations is exactly $s/2$.
    Observe that violations for $i>0$ have been created, but by
    Lemma~\ref{lemma:violations}, for each $i>0$, the number of $i$-violations
    remains $\leq s/2$.
  \item Case 2: the most frequent violations are $i$-violations for some $i > 0$.
    In that case, the first few times when a gadget is inserted, it may fail
    to kill the corresponding $i$-violation. But we claim that the number of
    such fails cannot be larger than $(\gamma+1)k+1$ (equivalently, in
    the algorithm the counter $d$ remains $\leq s/2+(\gamma+1)k+2$).

    Indeed, since we add a gadget only before an $i$-violation, the
    parsing splits the gadget $\ga{0}{j}= x_{<i} \bar x_i.1^j$ between
    $x_{< i}$ and $\bar x_i.1^j$. Furthermore, by induction,
    $\bar x_i.1^j$ is not split by the parsing. But for the gadget
    $\ga{i}{k+1}$, $\bar x_i.1^{k+1}$ is parsed in exactly one block
    because this factor does not appear anywhere in $w$ before
    $\ga{i}{k+1}$. From that moment on, each $i$-violation creates
    through the gadget a $0$-violation. The number of blocks that are
    both $0$-violated and $i$-violated is at most $i$ (due to the
    growth of the block at position $0$). Thus, at most $i$ more
    gadgets may fail to kill position $i$. The total number of ``failing''
    gadgets is $\leq k+1+i\leq (\gamma+1)k+1$.
  \end{itemize}
\end{proof}

\subsection{The weak catastrophe}
\label{sec:weak}

This section is devoted to the proof of the lower bound: the compression speed
of $0w$ is $\Omega(|w|^{3/4})$. Thanks to Property~\pstar,
Lemma~\ref{lem:junctions} below bounds the number of junction blocks ending at
a fixed position $(i-1)$ by a decreasing function of $i$. The proof is quite
technical and requires to distinguish three categories among (red) junction
blocks:
\begin{itemize}
\item Type 1: junctions over consecutive factors $w_a$ and $w_{a+1}$ (no gadget
  between two regular green blocks);
\item Type 2: junctions starting in a gadget
  $g_j^{j'}$ and ending in the following regular green block;
\item Type 3: junctions starting in a regular green block and ending in the
  following gadget $g_j^{j'}$.
\end{itemize}
\begin{lemma}
  \label{lem:junctions}
  Let $i\geq 2k+3$. Let $uu'$ be a junction block of type 1 over $w_aw_{a+1}$
  ending at position $i-1$ in $w_{a+1}$, with $u$ being the suffix of $w_a$
  and $u'$ the prefix of $w_{a+1}$. Then $|u| \leq k - \log (i - 2k - 1)$.

  In particular, the number of such blocks is upper bounded by the number of
  words of size $\leq k - \log (i-2k-1)$, that is, $\frac{2^{k+1}}{i-2k-1}$.
\end{lemma}
\begin{proof}
  Let $v$ be the prefix of size $2k$ of $u'$ (which is also the prefix of
  $x$). 
  All the prefixes of $uu'$ of size $\geq |uv|$ have to be in the dictionary
  of $0w$: we call $M$ the set of these prefixes ($|M|=i-2k$). We claim that
  these blocks are junction blocks of type 1 or 3 only (except possibly for one
  of type 2), with only $u$ on the left side of the junction. Indeed, let us
  review all the possibilities:
  \begin{enumerate}
  \item $uv$ cannot be completely included in a regular block, otherwise
    $v[0..k-1]$ would appear both at positions $0$ and $p>0$ in $x$, which
    contradicts Property~\pstar;
  \item $uv$ cannot be completely included in a gadget:
    \begin{itemize}
    \item if the gadget is $g_0^j=10^j$, impossible because $v$ cannot have
      more than $k$ zeroes since it is a factor of $x$,
    \item if the gadget is $g_b^j=x_{<b}.\bar x_b.1^j$, by the red parsing of
      gadgets, either $uv$ is in $x_{<b}$ (impossible because $v$ would appear
      at a position $\geq |u|$ in $x$), or $uv$ is in $\bar x_b.1^j$ (impossible
      because $v$ cannot contain more than $k$ ones);
    \end{itemize}
  \item if $uv$ is a type 1 junction but not split between $u$ and $v$, it is
    impossible because the three possible cases lead to a contradiction:
    \begin{itemize}
    \item if $u$ goes on the right, then $v$ would appear at another position
      $p>0$ in $x$,
    \item if $v$ goes on the left by at least $k$, then $v[0..k-1]$ would
      again appear at two different positions in $x$,
    \item if $v$ goes on the left by less than $k$, then it goes on the right
      by more than $k$ and $v[k..2k-1]$ would again appear at two different
      positions in $x$;
    \end{itemize}
  \item if $uv$ is a type 2 junction but not split between $u$ and $v$, it is
    again impossible:
    \begin{itemize}
    \item if $u$ goes on the right, then $v$ would appear at another position
      $p>0$ in $x$,
    \item if $v$ goes on the left by at least $2$, then $v[0..1]$ would be
      either $00$ or $11$ (depending on the gadget), but we know it is
      $x_0x_1=01$,
    \item otherwise, $v$ goes on the right by $2k-1$, and
      $v[1..k]$ would appear at positions $0$ and $1$ in $x$;
    \end{itemize}
  \item if $uv$ is a type 3 junction, first remark that the gadget is of the
    form $g_b^j$ for $b>0$ (because, for gadgets of the form $g_0^j$, the red
    parsing starts at position $0$ of the gadget). If $uv$ is not split
    between $u$ and $v$, it is once again impossible:
    \begin{itemize}
    \item if $u$ goes on the right, the red parsing of the gadget stops after
      $x_{<b}$ and $v$ would appear in $x$ at a non-zero position,
    \item similarly, if $v$ goes on the left by less than $k$, then
      $v[k..2k-1]$ would again appear at two different positions in $x$,
    \item if $v$ goes on the left by at least $k$, then $v[0..k-1]$ would
      again appear at two different positions in $x$.
    \end{itemize}
  \end{enumerate}
  Remark finally that all parsings of type 2 junctions have different sizes on
  the left. Therefore, at most one can contain $u$ on the left. The claim is
  proved.

  Thus, at least $|M|-1$ regular green blocks have $u$ as suffix. Remark that,
  since $|M|\geq 3$, there are at least two such green blocks, therefore
  $|u|\leq k$. Hence by Property~\pstar we have:
  \begin{align*}
    |M|-1 &\leq 2^{k-|u|}\\
    i-2k-1 &\leq 2^{k-|u|}\\
    |u| &\leq k - \log(i-2k-1).
  \end{align*}
\end{proof}

As a consequence, in the next proposition we can bound the size of offset-$i$
blocks. Along with the role of gadgets, this will be a key argument for the
proof of Theorem~\ref{thm:toy}. The idea is the following: for a red block $u$
starting at a sufficiently large position $i$, roughly $|u|$ other red blocks
have to end at position $(i-1)$, and in the red parsing $\Omega(|u|^2)$
prefixes of these blocks must appear in different green blocks (and in the
dictionary), giving the bound $s= \Omega(|u|^2)$.
\begin{proposition}\label{prop:smallblocks}
  For any $i > \gamma k$, the size of an offset-$i$ block included in a
  regular green block is at most
  $$2\sqrt{s} + 5k + \frac{2^{k+1}}{i-2k-1}.$$
\end{proposition}
\begin{proof}
  Let $u$ be an offset-$i$ block of size $\geq 2k$.

  We claim that the red blocks predecessors of $u$ of size at
  least $2k+1$
  have to start at position $i$ in regular green blocks. Indeed, let $v$ be a
  prefix of size $\geq 2k+1$ of $u$; let us analyse as before the
  different cases:
  \begin{itemize}
  \item If $v$ is included in a regular green block, then it has to start at
    position $i$ by Property~\pstar;
  \item $v$ cannot be included in a gadget since it would lead to a contradiction:
    \begin{itemize}
    \item in gadgets of type $g_0^j$, $v$ would contain $0^{2k}$,
    \item in gadgets of type $g_a^j=x_{<a}\bar x_a 1^j$ (for
      $a\in]0,\gamma k]$), either $v$ goes into $x_{<a}$ by at least $k$ and
      $v[0..k-1]$ would appear at two positions in $x$, or $v$ goes into $\bar
      x_a 1^j$ by at least $k+2$ and $v$ would contain $1^{k+1}$;
    \end{itemize}
  \item If $v$ is included in a junction block of type 1, then $v$ starts at
    position $i$ in the left regular block, otherwise either $v[0..k-1]$ would
    be in the left regular block at a position different from $i$, or
    $v[k+1..2k]$ would be in the right regular block at a position
    $\leq \gamma k<i$;
  \item $v$ cannot be included in a junction block of type 2: indeed, by the
    red parsing, the left part of the junction (included in a gadget) is
    either $10^j$ or $a1^j$ for some letter $a\in\{0,1\}$, thus $v$ cannot go
    on the left by $\geq k+2$ and hence has to go on the right by at least $k$
    leading to a contradiction with Property~\pstar;
  \item If $v$ is included in a junction block of type 3, then $v$ starts at
    position $i$ in the left regular block, otherwise either $v[0..k-1]$ would
    be in the left (regular) block at a position different from $i$, or, by the
    red parsing, the gadget is of type $g_a^j$ (for $a>0$) and $v[k+1..2k]$
    would be included in $x_{<a}$ at a position $\leq \gamma k<i$.
  \end{itemize}
 Thus, at least
  $|u| - 2k$ red blocks end at position $i-1$.
  
  By Lemma \ref{lem:junctions}, at most $\frac{2^{k+1}}{i-2k-1}$ of them are
  junctions of type 1. Note furthermore that, as shown during the proof of
  Lemma~\ref{lemma:fewviolations}, if $a\geq k+1$, the $a$-th junction of type
  2 stops at position $0$ or $1$, hence at most $k$ of the blocks ending at
  position $i-1$ are junctions of type 2.
  Finally, there is, by definition, no junction of type 3. Therefore, there
  are at least $|u|-3k-\frac{2^{k+1}}{i-2k-1}$ offset blocks ending at
  position $i-1$. We call $M$ the set of such blocks. See
  Figure~\ref{fig:antichains}. Remark that $|u|-3k-\frac{2^{k+1}}{i-2k-1}$ is
  a lower bound on the number of offset blocks ending at position $i-1$. But
  the number of such blocks is at most $i$. Therefore
  $$|u|-3k-\frac{2^{k+1}}{i-2k-1} \leq i$$
  We distinguish two cases in the proof:

  \textbf{First case:} $i\in [\gamma k +1, 2\sqrt{s}]$.  
  Then
  $$|u|-3k-\frac{2^{k+1}}{i-2k-1} \leq i \leq 2\sqrt{s}$$
so that
  $$|u| \leq 2\sqrt{s} + 3k + \frac{2^{k+1}}{i-2k-1} \leq 2\sqrt{s} + 5k + \frac{2^{k+1}}{i-2k-1}.$$

\textbf{Second case:} $i > 2\sqrt{s}$.

All the words in $M$ are in the dictionary and are of different size,
since two offset blocks ending at the same position and of same size
would be identical, which is not possible in the LZ-parsing. The words
of $\p(M)$ (the set of prefixes of the words in $M$) are also in the
dictionary. Let
  $$A = \frac{|u| - 5k - \frac{2^{k+1}}{i-2k-1}}{2}.$$

  Observe that $i-A-2k \geq \frac{i}{2} - k$ as
  $|u|-3k-\frac{2^{k+1}}{i-2k-1} \leq i$. Therefore
  $i-A-2k \geq \sqrt{s} - k$, which is large against $\gamma
  k$.
  Consider the words of $\p(M)$ containing $x[i-A-2k..i-A]$: they
  must start at a position $\leq i-A-2k$ and end at a position
  $\in [i - A, i-1]$. The number of such words is at least the
  product of the number of blocks in $M$ starting at position
  $\leq i-A - 2k$ and of the number of possible ending points, that
  is, at least
  $$\biggl(\bigl(|u| - 3k - \frac{2^{k+1}}{i-2k-1}\bigr) - \bigl(A + 2k\bigr)\biggr)A.$$
  Remark that these words contain a part of a regular green block of
  size at least $2k+1$ starting at position $i-A-2k > \gamma k$.
  Hence, by the same case analysis as before, for these words, the
  part corresponding to the factor $x[i-A-2k..i-A-k-1]$ must appear
  included in a regular green block, so that two such words cannot
  appear in the same regular green block by Property~\pstar.  But
  there are at most $s$ distinct regular green blocks, thus:
  $$\left(|u| - 5k - \frac{2^{k+1}}{i-2k-1} - A\right)A \leq s.$$
  The value of $A$ gives:
  \begin{align*}
    \left(\frac{|u| - 5k - \frac{2^{k+1}}{i-2k-1}}{2}\right)^2 &\leq s\\
    |u| &\leq 2\sqrt{s} + 5k + \frac{2^{k+1}}{i-2k-1}.
  \end{align*}
\end{proof}

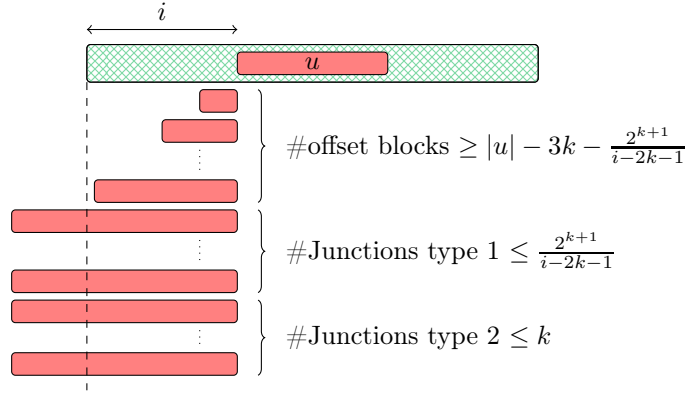
\begin{figure}[h]
  \centering
  \begin{tikzpicture}[scale=1, every node/.style={scale=1}]

  \draw [<->] (0,0.7) --node[above] {$i$} (2,0.7);

  \draw [pattern=north east lines,pattern color = green!70!blue!50,rounded corners=0.5mm] (0,0) rectangle (6,0.5);
  \draw [pattern=north west lines,pattern color = green!70!blue!50,rounded corners=0.5mm] (0,0) rectangle (6,0.5);

  \draw [fill=red!50,rounded corners=0.5mm] (2,0.1) rectangle (4,0.4) node[pos=.5] {$u$};

  \draw [fill=red!50,rounded corners=0.5mm] (1.5,-0.4) rectangle (2,-0.1); 
  \draw [fill=red!50,rounded corners=0.5mm] (1,-0.8) rectangle (2,-0.5); 

  \draw [dotted] (1.5,-0.9) -- (1.5,-1.2);

  \draw [fill=red!50,rounded corners=0.5mm] (0.1,-1.6) rectangle (2,-1.3);

  \draw [fill=red!50,rounded corners=0.5mm] (-1,-1.7) rectangle
  (2,-2);

  \draw [dotted] (1.5,-2.1) -- (1.5,-2.4);

  \draw [fill=red!50,rounded corners=0.5mm] (-1,-2.5) rectangle
  (2,-2.8);

  \draw [fill=red!50,rounded corners=0.5mm] (-1,-3.2) rectangle
  (2,-2.9);

  \draw [dotted] (1.5,-3.3) -- (1.5,-3.5);

  \draw [fill=red!50,rounded corners=0.5mm] (-1,-3.6) rectangle
  (2,-3.9);

  \draw [dashed] (0,0) -- (0,-4.1);

  \draw[decoration={brace,raise=5pt},decorate]  (2.1,-0.1) --
    node[right=12pt] {\#offset blocks $\geq |u|-3k - \frac{2^{k+1}}{i-2k-1}$} (2.1,-1.6);

  \draw[decoration={brace,raise=5pt},decorate]  (2.1,-1.7) --
    node[right=12pt] {\#Junctions type 1 $\leq \frac{2^{k+1}}{i-2k-1}$} (2.1,-2.8);

 \draw[decoration={brace,raise=5pt},decorate]  (2.1,-2.9) --
    node[right=12pt] {\#Junctions type 2 $\leq k$} (2.1,-3.9);

  \end{tikzpicture}
  \caption{Blocks ending at position $i-1$ for the proof of
    Proposition~\ref{prop:smallblocks}.}
  \label{fig:antichains}
\end{figure}

We are ready for the proof of the main theorem of this section.
\begin{proof}[Proof of Theorem~\ref{thm:toy}]

  The intuition is the following: by Proposition~\ref{prop:smallblocks}, the
  red blocks starting at position $j$, for $j = \Omega(\sqrt{s})$, are of size
  $\Theta(\sqrt{s}) = \Theta(|w|^{1/4})$, so if we prove that a portion of
  size $\Theta(|w|)$ of the word $0w$ is covered by offset-$j$ blocks for $j$
  large enough, then the compression speed will be $\Omega(|w|^{3/4})$. To
  that purpose, we prove that for large enough regular green blocks, there is
  an interval of positions $[2s/3-l,2s/3]$ (with $l = 2\sqrt{s} + 5k + 3$),
  such that there is at least one offset-$i$ block for $i\in [2s/3-l,2s/3]$.

  In every regular green block of size larger than $2s/3$, let us show that
  there is an offset-$i$ red block, for $i \in [2s/3-l,2s/3]$. Indeed, for
  every $i < 2s/3 - l$, the maximal size $f(i)$ of a red block starting at
  position $i$ satisfies $i + f(i) \leq 2s/3$: in the case where $i> \gamma k$
  we use the bound given by Proposition~\ref{prop:smallblocks}, and in the
  case where $i\leq\gamma k$, the $(t-\gamma k)$ predecessors of size
  $\geq \gamma k+1$ of a red block of size $t$ starting at position $i$ start
  at position $i$ as well (since $x_{\leq\gamma k}$ is not a factor of a
  gadget, it cannot be seen anywhere by a red block except at position $i$ in
  a regular green block), hence the size of an offset-$i$ block in that case
  is at most $\gamma k$ plus the number of $i$-violations. Therefore by
  Lemma~\ref{lemma:fewviolations} red blocks starting at position $i$ have
  their size upper bounded by $\gamma k+s/2 +1 + (1 + \gamma)k$.

  Therefore, since a red block starting at position $i \geq 2s/3-l$ is of size
  at most $B=2\sqrt{s}+5k+\frac{2^{k+1}}{2s/3-l-2k-1}$ by
  Proposition~\ref{prop:smallblocks}, each green block of size $h \geq 2s/3$
  is covered by at least
  $$\frac{h-2s/3}{B} \geq \frac{h-2s/3}{2\sqrt{s} + O(k)}$$
  red blocks. Thus, the total number of red blocks is at least
  $$\frac{1}{2\sqrt{s} + O(k)}\sum_{h=2s/3}^s (h-2s/3) =
  \frac{1}{36}s^{3/2} + o(s^{3/2}).$$
  With the gadgets, the size of $w$ is at most $(5/8)s^2 + o(s^2)$,
  therefore the total number of red blocks is at least:
  $$\frac{1}{36}\left(\frac{8}{5}\right)^{3/4} |w|^{3/4}-o(|w|^{3/4})\geq
  0.039|w|^{3/4}.$$
\end{proof}

\begin{remark}\label{rem:gadgets}
  Despite the fact that $1\pref(x)$ compresses optimally, this is not
  at all the case with the gadgets, since Theorem~\ref{thm:toy}
  remains valid with the new word $w$ output by the algorithm even
  when we put $1$ instead of $0$ in front of $w$.
\end{remark}

\section{General case}
\label{sec:general}

In this section we prove Theorem~\ref{thm:lowerbound}. The proof first goes
through the existence of a family $F$ of ``independent'' de Bruijn-style words
which will play a role similar to the de Bruijn word $x$ in the proof of
Theorem~\ref{thm:toy}. The existence of this family is shown using the
probabilistic method in Section~\ref{sec:family-de-bruijn}: with high
probability, a family of random words satisfies a relaxed version (P1) of the
``local'' Property~\pstar, together with a global property (P2) that forbids
repetitions of large factors throughout the whole family.

The word $w$ that we will consider is the concatenation of ``chains'' roughly
equal to $\pref(x)$ for all words $x\in F$, with gadgets inserted if necessary
as in Section~\ref{sec:toy}. (The construction is actually slightly more
complicated because in each chain we must avoid the first few prefixes of $x$
in order to synchronise the parsing of $w$; and the gadgets are also more
complex.) Properties~(P1) and~(P2) guarantee that each of the chains of $w$
are ``independent'', so that the same kind of argument as in
Section~\ref{sec:toy} will apply individually. By choosing appropriately the
number of chains and their length, we can obtain any compression speed for $w$
up to $\Theta(n/\log^2n)$ and the matching bound for $0w$ (see
Theorem~\ref{thm:lowerbound}).

The organisation of the section is as follows:
Section~\ref{sec:family-de-bruijn} is devoted to the proof of existence of the
required family of words. Section~\ref{sec:construction} defines the gadgets,
describes the construction of $w$ thanks to an algorithm, and gives the upper bound
on the compression speed of $w$. Finally, Section~\ref{sec:proofs} shows the
lower bound on the compression speed of $0w$ thanks to a series of results in
the spirit of Section~\ref{sec:weak}.

Throughout the present section, we use parameters with some relations between
them that are worth being stated once and for all in
Figure~\ref{fig:parameters} for reference.

\begin{figure}[h]
  \centering
  \boxed{\def\arraystretch{1.6}
    \begin{array}{l}
      n \text{ sufficiently large (the size of $w$)}\\
      \gamma \geq 10 \text{ (an absolute constant)}\\
      l\in[(9\gamma)^2\log^2n, \sqrt{n}] \text{ (size of the $x^j$)}\\
      p = \log\frac{n}{l^2} \text{ ($2^p$ is the number of chains)}\\
      k = \frac{\log l}{2} \text{ (parameter in (P1))}\\
      m =\max(\gamma p, \gamma\log l) \text{ (parameter in (P2))}.
    \end{array}}
  \caption{Parameters used throughout Section~\ref{sec:general}.}
  \label{fig:parameters}
\end{figure}

In particular, note that we have the following relations:
$$0\leq p\leq\sqrt{l}\quad\text{and}\quad \frac{\gamma\log n}{3}\leq m\leq\frac{\sqrt{l}}{9}.$$

\subsection{Family of de Bruijn-type words}
\label{sec:family-de-bruijn}

We need two properties for a family $F$ of $2^p$ words $x^1,\dots,x^{2^p}$ of
size $l$ (the parameters $n$, $l$, $p$, $k$, $\gamma$ and $m$ are those given
in Figure~\ref{fig:parameters}): the first is a relaxed version of
Property~\pstar on ``true'' de Bruijn words; the second guarantees that the
words of $F$ are ``independent''.
\begin{itemize}
\item (P1) For all $x\in F$, for all words $u$ of size $\leq k$,
  $$\occ_x(u)\leq\frac{kl}{2^{|u|}}.$$
\item (P2) Any factor $u$ of size $m$
  appears in at most one word of the family $F$, and within that word at
  only one position.
\end{itemize}
Note that in Section~\ref{sec:toy}, we did not need (P2) since only one word
was concerned, but still (P2) was true for the same value $k$ as in~\pstar,
instead of $m$ here.

The following lemmas show that (P1) and (P2) hold with high probability
for a random family $F$. We first recall the well-known Chernoff bound.

\begin{theorem}[Chernoff bound]\label{thm:chernoff}
  Let $X_1,\dots,X_n$ be independent random variables over $\{0,1\}$, and
  $X=\sum X_i$. Denote by $\mu$ the expectation of $X$. Let $\delta>1$. Then:
  $$\pr(X>\delta\mu)
  <2^{-\frac{(\delta-1)\mu\log\delta}{2}}.$$
\end{theorem}

For (P1), we need to consider positions separated by a distance $k$ in order
to obtain the independence required for the Chernoff bound; then a union bound
will complete the argument for the other positions.
\begin{lemma}[(P1) holds whp]
  Let $p$ and $l$ be positive integers such that $p\leq\sqrt{l}$. Let $F$ be
  a family of $2^p$ words $x^1,\dots,x^{2^p}$ of size $l$ chosen uniformly and
  independently at random. Then $F$ satisfies Property~(P1) with probability
  $2^{-\Omega(\sqrt{l}\log\log l)}$.
\end{lemma}

\begin{proof}
  Fix $x\in F$ and a word $u$ of size $\leq k$. For $i\in[0,k-1]$ and
  $j\in[0,l/k-1]$, let
  $$X_j^i=\left\{
    \begin{array}{ll}
      1 & \text{if $u$ occurs at position $i+jk$ in $x$}\\
      0 & \text{otherwise.}
    \end{array}\right.$$
  For a fixed $i$, the $X_j^i$ are independent. Let $\mu_i=E(\sum_j X_j^i)$. We have:
  $$\mu_i=\frac{l}{k2^{|u|}}.$$
  By the Chernoff bound (Theorem~\ref{thm:chernoff}):
  $$\pr\biggl(\sum_j X_j^i > k\mu_i\biggr) < 2^{-\frac{(k-1)(\log k)\mu_i}{2}},$$
  that is,
  $$\log\pr\biggl(\sum_j X_j^i > \frac{l}{2^{|u|}}\biggr) < -(k-1)(\log k)\frac{l}{k2^{k+1}}.$$
  By union bound over all the words $u$ of size at most $k$, all the words of
  $F$ and all the moduli $i\in[0,k-1]$, we have:
  \begin{align*}
  \log\pr\biggl(\sum_{i,j} X_j^i > \frac{kl}{2^{|u|}}\biggr) & < (k+1)+p+\log k-(k-1)(\log
  k)\frac{l}{k2^{k+1}}\\
    & =-\Omega\bigl(\sqrt{l}\log\log l\bigr).
  \end{align*}
\end{proof}

The analysis for (P2) does not use Chernoff bounds, but instead it uses a
slight ``independence'' on the occurrences of a factor $u$ obtained by showing
that $u$ can be supposed ``self-avoiding'' (the precise meaning of these ideas
will be clear in the proof).

\begin{lemma}[(P2) holds whp]\label{lem:p2whp}
  Let $p$ and $l$ be positive. Recall that $m=\max(\gamma p,\gamma\log l)$ in
  (P2). Let $F$ be a family of $2^p$ words $x^1,\dots,x^{2^p}$ of size $l$
  chosen uniformly and independently at random. Then $F$ satisfies
  Property~(P2) with probability a least
  $1-2/l$.
\end{lemma}

\begin{proof}
  Let us first show that we can assume with high probability that factors of
  $w_i$ are not overlapping too much. We say that a word $u$ of size
  $m$ is ``bad'' if it overlaps itself at least by half, that is:
  $$\exists i\in [|u|/2, |u|-1]: u[0..i-1]=u[|u|-i..|u|-1]\text{
    (we say that $u$ is $i$-bad)}.$$
  (Remark that a word $u$ can be both $i$-bad and $j$-bad for $i\neq j$.) Let
  us first bound the number of bad words. If $u$ is $i$-bad, then for each
  $j<i$, $u_{j+|u|-i}=u_j$. Therefore, specifying the $|u|-i$ first bits
  specifies the whole word $u$, meaning that there are at most $2^{|u|-i}$
  $i$-bad words. In total, there are at most
  $$\sum_{i=|u|/2}^{|u|-1}2^{|u|-i}=2^{1+|u|/2}-2$$
  bad words, that is, a fraction $<2^{-m/2+1}$ of all words of size $m$.

  Now, we say that a word $x^j$ of size $l$ is ``good'' if it contains no bad
  factor. Let us show that, with high probability, all the words $x^j\in F$
  are good (Property~(G)). Fix $j\in [1,2^p]$. If $x^j$ is not good, then
  there is at least one position where a bad factor $u$ occurs:
  $$\pr(x^j \text{ is not good})\leq |x^j|\pr_{|u|=m}(u\text{ is
    bad})\leq l2^{-m/2+1}.$$
  We use the union bound over all $2^p$ words $x^j\in F$ to obtain:
  $$\pr(G)\geq 1-l2^{p-m/2+1}.$$

  Since property $G$ has very high probability, we will only show that
  (P2) holds with high probability when $G$ is satisfied. Let
  $x = x^1\dots x^{2^p}$ (the size of $x$ is therefore $l2^p$). Let
  $u$ be a word of size $m$, which is not bad. Let $X_u$ be the number
  of occurences of $u$ in $x$. In order to get at least two
  occurrences of $u$, we have to choose two positions, the $|u|$ bits
  of the first occurrence, and the bits of the second occurrence that
  are not contained in the first; but $u$ can't overlap itself by more
  than $m/2$ bits, thus:
  $$\pr(X_u \geq 2) \leq \frac{|w|}{2^{m}}
  \frac{|w|}{2^{m/2}} \leq l^22^{2p-\frac{3}{2}m}.$$
  Using the union bound over all good words $u$ of size $m$, of which
  there are at most $2^{m}$, we get:
  $$\pr(\forall\ \mathrm{good}\ u, X_u \leq 1) \geq
  1-2^{m}l^22^{2p-\frac{3}{2}m} = 1-l^2 2^{2p-m/2}.$$

  Now, the probability that $F$ respects Property (P2) can be lower bounded by
  the probability that $F$ contains no bad words, and that the number of
  occurences of good words is at most $1$, which gives:
  \begin{align*}
  \pr(F \text{ satisfies (P2)}) &\geq \pr(G\land\forall\ \mathrm{good}\ u,
  X_u \leq 1)\\
    &\geq 1-l^22^{2p-m/2} - l2^{p-m/2+1}\\
    &>1-\frac{2}{l}\text{ since }\gamma\geq 10
  \end{align*}
  (for the last line, consider the two following cases: $p\geq\log l$ where
  $m=\gamma p$ and $l\leq n^{1/3}$; and $p\leq\log l$ where $m=\gamma\log l$
  and $l\geq n^{1/3}$).
\end{proof}

\begin{corollary}\label{cor:family}
  For all sufficiently large $l$ and $p\leq\sqrt{l}$, there exists a
  family $F$ of $2^p$ words $x^1,\dots,x^{2^p}$ of size $l$ satisfying
  Properties~(P1) and~(P2), and where the first bit of $x^1$ is $1$.
\end{corollary}

\subsection{Construction}
\label{sec:construction}

(Recall the choice of parameters $n$, $l$, $p$, $k$, $\gamma$ and $m$ defined in
Figure~\ref{fig:parameters}.)

For $n$ sufficiently large and $l\in[(9\gamma)^2\log^2n,\sqrt{n}]$, we
will construct a word $w$ of size $n$ whose compression speed is $\Theta(n/l)$
whereas the compression speed of $0w$ is $\Theta(n/\sqrt{l})$ (thus matching
the upper bound of Theorem~\ref{thm:upperbound}). Let $F$ be a family as in
Corollary~\ref{cor:family}. For some integers $q_j$ (defined below), the word
$w$ will merely be the concatenation of $\pref_{>q_j}(x^j)$ (see
Remark~\ref{rem:extremal} for the definition of $\pref_{>q}(x)$) for all the
$2^p$ words $x^j$ of the family $F$, with possibly some gadgets added between
the prefixes of $x^j$ (each $\pref_{>q_j}(x^j)$ together with the possible
gadgets will be denoted $z^j$ and called a ``chain''), and a trailing set of
zeroes so as to ``pad'' the length to exactly $n$. The integer $q_j$ will be
chosen so that the first occurrence of $x^j[0..q_j]$ is parsed in exactly one
green block.

Each chain $z^j$ (with gadgets) is of size $\Theta(l^2)$ and is fully
compressible in $w$ (compressed size $\Theta(l)$) since it is made of prefixes
(plus gadgets that won't impede much the compression ratio). Thus the total
compression size of $w$ is $\Theta(l2^p)$, compared to $|w|=n=\Theta(l^22^p)$
for a compression speed of $\Theta(n/l)$.

On the other hand, due to the properties of $F$ and similarly to
Theorem~\ref{thm:toy}, in $0w$ each chain will compress only to a size
$\Theta(l^{3/2})$, thus the total compression size of $0w$ is
$\Theta(l^{3/2}2^p)$, for a compression speed of $\Theta(n/\sqrt{l})$.

\begin{remark}
  \begin{itemize}
  \item If we take the smallest possible $l$, that is,
    $l=(9\gamma)^2\log^2n$, then we obtain compression speeds of
    $\Theta(n/\log^2n)$ and $\Theta(n/\log n)$, thus showing the one-bit
    catastrophe.
  \item On the other hand, if we take the largest possible $l$, that is,
    $l=\sqrt{n}$, then we obtain $\Theta(\sqrt{n})$ and $\Theta(n^{3/4})$ as in
    Theorem~\ref{thm:toy}.
  \end{itemize}
\end{remark}

Let us now start the formal description of the word $w$. As previously, we
will call green the blocks in the parsing of $w$ and red those in the parsing
of $0w$. The green blocks in each chain $z^j$ that are not gadgets will be
called ``regular blocks'' (they are of the form $x^j[0..q]$ for some $q$).
Recall that the chain $z^j$ will be of the form $\pref_{>q_j}(x^j)$ with
possibly some gadgets between the prefixes. We can already define the integers
$q_j$:
$$q_j=\min\{i\geq 0:x^j[0..i]\text{ is not a prefix of }x^1,\dots,x^{j-1}\}.$$
In that way, we guarantee that the first green block in each $z^j$ is exactly
$x^j[0..q_j]$. Remark that, by Property~(P2), $q_j\in[0,m]$. For all $j$ we
will denote by $s_j=|x^j|-q_j=l-q_j$ the number of regular green blocks in
$z^j$.

Fix $n$ and $l=l(n)\in[(9\gamma)^2\log^2n,\sqrt{n}]$, and let $k=(\log l)/2$
and $p=\log(n/l^2)$. As in Property~(P2), call
$m=\max(\gamma p,\gamma\log l)$. Here are the new gadgets that will (possibly)
be inserted in the chain $z^j$ ($j\in[1,2^p]$), for $i\in[0,2k\sqrt{l}]$:
\begin{itemize}
\item for $c\geq 0$: $g_0^c(j) = ua^c$, where $a=x^j[0]$ is the first letter
  of $x^j$, and $u$ is the smallest word in
  $$\dic(0z^1\dots z^{j-1}\pref_{>q_j}(x^j[0..|x^j|/2]))$$ but not in
  $$\dic(z^1\dots z^{j-1}\pref_{>q_j}(x^j)):$$ this is a word which is in the parsing
  of $0w$ up to the insertion of $g_0^0(j)$ but not in the corresponding
  parsing of $w$ (Lemma~\ref{lem:existence-word} below guarantees the
  existence of such a word and proves it is of size $\leq m$);
\item for $i>0$ and $c\geq 0$, let $m'=\max(i,m)$ and $v=x^j[0..m-1]1^l$. Then:
  $$g_i^c(j) = x^j[0..m'-1]\bar x^j_{m'}v[0..c-1]$$
  where $\bar x^j_{m'}$ denotes the complement of $x^j[m']$.
\end{itemize}

We define $i$-violations in each chain $z^j$ as previously, that is, a regular
green block is $i$-violated if it contains an offset-$i$ red block. The
following lemma is proved in the exact same way as
Lemma~\ref{lemma:violations}.
\begin{lemma}\label{lemma:violations2}
  For $j\in[1,2^p]$ and $i\in[0,s_j-1]$, let $l_i^j$ be the number of
  $i$-violated blocks in $z^j$. Then for
  all $j$ and all $i\neq i'$, $l_i^j+l_{i'}^j\leq s_j$.

  In particular, for each $z^j$ there can be at most one $i$ such that the
  number of $i$-violated blocks is $>s_j/2$.
\end{lemma}

The formal construction of the word $w$ is once again best described by an
algorithm taking as parameters $n$ and $l$:
\begin{enumerate}
\item For all $j\in[1, 2^p]$ and $i\in[q_j+1,l]$, $z_i^j\leftarrow x^j[q_j..i-1]$.
  Throughout the algorithm, $z^j$ will denote $z_{q_j+1}^j\dots z_l^j$ (and thus
  will vary if one of the $z_i^j$ varies).
\item For $j=1$ to $2^p$ do:
  \begin{enumerate}
  \item if there is $i\in[0,2k\sqrt{l}]$ (unique by Lemma~\ref{lemma:violations2}) such
    that the number of $i$-violations in the chain $z^j$ is $>s_j/2$, then:
    \begin{enumerate}
    \item let $c=0$ (counter for the number of inserted gadgets in $z^j$) and
      $d=s_j/2+1$ (counter for the place of the gadget to be inserted),
    \item while the number of $i$-violations in the chain $z^j$ is $\geq d$, do:
      \begin{enumerate}
      \item let $r$ be such that $z^j_r$ is the $d$-th $i$-violated green
        block in $z^j$,
      \item $z^j_r\leftarrow g_i^{c}(j)z^j_r$ (we add the gadget $g_i^c(j)$
        before the block $x^j[0..r-1]$),
      \item $c\leftarrow c+1$,
      \item if $z^j_r$ is still $i$-violated, then $d\leftarrow d+1$.
      \end{enumerate}
    \end{enumerate}
  \end{enumerate}
\item Let $w'=z^1.z^2\dots z^{2^p}$. Return $w=w'0^{n-|w'|}$ (padding to
  obtain $|w|=n$).
\end{enumerate}

Remark that we have the following bounds on the size of $w'$. Its size is
minimal if no gadgets are added:
$$|w'|\geq 2^p\biggl(\sum_{i=m}^l i\biggr) =
\frac{n}{l^2}\cdot\frac{(l-m+1)(l+m)}{2}\geq \frac{n}{2} - o(n)$$
and its size is maximal if each chain contains $l/2$ gadgets $g_i^c$
(whose size is at most $2k\sqrt{l}+1+c$):
\begin{align*}
  |w'|&\leq 2^p\biggl(\sum_{i=1}^li+\sum_{c=0}^{l/2-1}(2k\sqrt{l}+1+c)\biggr)\\
      &\leq\frac{n}{l^2}\biggl(\frac{5l^2}{8}+2kl^{3/2}\biggr)\leq n.
\end{align*}
Therefore at the end of the algorithm it is legitimate to
pad $w'$ with at most $n/2+o(n)$ zeroes to obtain the word $w$ of size precisely
$n$.

The following lemma justifies the existence of the gadgets $g_0^c(j)$.
\begin{lemma}\label{lem:existence-word}
  There is a constant $C>0$ such that, for all $n$, for all
  $l\in[(9\gamma)^2\log^2 n,\sqrt{n}]$ with $l>C$, for all $j\in[1,2^p]$
  (where $p=\log(n/l^2)$) there exists a word $u$ of size $\leq m$ in
  $$\dic(0z^1\dots z^{j-1}\pref_{>q_j}(x^j[0..|x^j|/2]))$$ but not in
  $\dic(z^1\dots z^{j-1}\pref_{>q_j}(x^j))$.

  This implies that we can insert the gadgets $g_0^c(j)$ in a chain $z^j$
  whenever we need to.

\end{lemma}
\begin{proof}
  For $j=1$: the first red block in $z^1$ is $0$, but all the
  regular green blocks in $z^1$ begin with $1$ (cf.
  Corollary~\ref{cor:family}). Therefore there exists a word $u$ of size
  $1$ in $$\dic(0\pref_{>q_1}(x^1[0..|x^1|/2]))$$ not in
  $\dic(\pref_{>q_1}(x^1))$.

  For $j>1$: as we shall see in the proof of Theorem~\ref{thm:lowerbound}
  below\footnote{That is not a circular argument because we only need the
    result up to $z^{j-1}$ to claim the existence of gadgets for $z^j$.}, in
  $z^{j-1}$ there is an interval of $A=3\sqrt{l}$ positions that contains
  the starting position of a red block in at least $l/3$ regular green blocks.
  One of the
  positions of the interval is the starting point of a branch of at least
  $$\frac{l}{3A}= \frac{\sqrt{l}}{9}\geq m$$ red
  blocks. Thus there is a red block of size $m$ which appears nowhere in
  $z^1,\dots,z^{j-1}$ by (P2) and is not a prefix of $z^j$, thus it is a word
  of $\dic(0z^1\dots z^{j-1})$ not in
  $\dic(z^1\dots z^{j-1}\pref_{>q_j}(x^j))$.
\end{proof}

We can now show that $w$ has compression speed $O(|w|/l)$ by giving an upper
bound on the size of the dictionary of $w$.
\begin{lemma}\label{lem:general-upper-bound}
  The compression speed $|\dic(w)|$ of $w$ is at most
  $$\frac{3+\sqrt{3}}{2}\cdot\frac{|w|}{l}.$$
\end{lemma}

\begin{proof}
  The definition of the integers $q_j$ guarantees that the parsing
  resynchronizes at each beginning of a new chain $z^j$.

  In a chain $z^j$, the definition of $g_0^0(j)$ guarantees that this gadget,
  if present, will be parsed in exactly one green block, and after that the
  subsequent gadgets $g_0^c(j)$ also.

  Similarly, (P2) together with the fact that gadgets are only inserted in the
  second half of a chain (thus, after more than $m$ green blocks) imply that
  the possible gadgets $g_i^c(j)$ for $i>0$ are also parsed in exactly one
  green block.

  For each chain $z^j$, the parsing tree consists in a main path of size $l$
  (regular green blocks) together with another path of size $\leq l/2$
  corresponding to the gadgets $g_i^c(j)$. The compression speed cannot be
  worse than in the (hypothetical) case where these two paths begins at depth
  $0$, for all $j$. In that case, there are $\leq (3/2)l$ green blocks for
  each chain, and a size
  $$|z^j|\geq \frac{l(l+1)}{2}+\frac{\frac{l}{2}(1+\frac{l}{2})}{2}\geq
  \frac{5}{8}l^2.$$
  Since the number of chains is $2^p=n/l^2$, in that (hypothetical) worst
  case the number of green blocks in $w'$ is at most $3n/2l$ and
  $|w'|\geq 5n/8$. The $\leq 3n/8$ trailing zeroes of $w$ are parsed in at most
  $\sqrt{3n}/2\leq (\sqrt{3}/2)n/l$ green blocks. Hence the compression speed of
  $w$ is at most $(3/2+\sqrt{3}/2)(n/l)$.
\end{proof}

\subsection{Proof of the main theorem}
\label{sec:proofs}

(Recall the choice of parameters $n$, $l$, $p$, $k$, $\gamma$ and $m$ defined in
Figure~\ref{fig:parameters}.)

We now prove the lower bound of Theorem~\ref{thm:lowerbound}. Recall that
$z^j$ denote the $j$-th chain of $w$. We will write $w_{i}^j$ the $i$-th
regular block of the chain $z^j$. As in the previous section, we will
distinguish junctions over two consecutive regular blocks (type~1); junctions
starting in a gadget and ending in a regular block (type~2); and junctions
starting in a regular block and ending in a gadget (type~3).

The next proposition is the core of the argument, and
Theorem~\ref{thm:lowerbound} will follow easily. The proposition is a
corollary of lemmas that we will show afterwards.

\begin{proposition}
  \label{proposition:size-blocks}
  Let $f(i)$ be the maximal size of an offset-$i$ (red) block included in a
  regular green block.
  \begin{itemize}
  \item If $i \leq 2k\sqrt{l}$ then $f(i) \leq \frac{l}{2} + 4k\sqrt{l} + 2m+ 1$.
  \item Otherwise, $f(i) \leq 2\sqrt{l} + 3k+ 7m + \frac{2kl}{i - 4m - 2}$.
  \end{itemize}
\end{proposition}
\begin{proof}
  The first point is a consequence of Lemmas~\ref{lemma:number-violations}
  and~\ref{lemma:violation-size}. The second point is exactly
  Lemma~\ref{lemma:size-offset}.
\end{proof}

With Proposition~\ref{proposition:size-blocks} in hand, let us prove the main
theorem.
\begin{proof}[Proof of Theorem~\ref{thm:lowerbound}]
  We will show that each chain $z^j$ in $0w$ is parsed in at least
  $\frac{1}{54} l^{3/2}$ blocks, thus
  $$|\dic(0w)|\geq 2^p \frac{1}{54} l^{3/2} =
  \frac{1}{54}\cdot\frac{|w|}{\sqrt{l}}.$$

  Fix an index $j$. In order to prove that the chain $z^j$ is parsed in at
  least $\frac{1}{54}l^{3/2}$ red blocks, we first prove that in every regular
  green block of size larger than $2l/3$ in the chain $z^j$, there is an
  interval of positions $[2l/3-A,2l/3]$ (with $A = 3\sqrt{l}$), such
  that there is at least one offset-$i$ (red) block for $i \in [2l/3-A,2l/3]$.
  Indeed, by Proposition~\ref{proposition:size-blocks}, for any $i<2l/3-A$,
  the maximal size $f(i)$ of a red block starting at position $i$ satisfies
  $i + f(i) \leq 2l/3$.

  Therefore, since the red blocks starting at position $i \geq 2l/3-A$ are of
  size at most $f(2l/3-A)\leq 3\sqrt{l}$, a regular green block of $z^j$ of
  size $h$ is covered by at least $(h-2l/3)/(3\sqrt{l})$ red blocks. Thus the
  number of red blocks in the parsing of $z^j$ is at least
  $$\sum_{h=2l/3}^l \frac{h-2l/3}{3\sqrt{l}} \geq
  \frac{1}{54}l^{3/2}.$$
\end{proof}

Now we prove Proposition~\ref{proposition:size-blocks} thanks to the next four
lemmas. The first two show that the gadgets do their job: indeed, for small
$i$ (the indices $i$ covered by the gadgets), the offset-$i$ blocks are not
too large. For that, we first bound the number of violations.
\begin{lemma}
\label{lemma:number-violations}
  For any $i\in [0,2k\sqrt{l}]$ and $j\in[1,2^p]$, the number of $i$-violations
  in the chain $z^j$ is at most $\frac{l}{2} + 2m + 1 + 2k\sqrt{l}$. 
\end{lemma}
\begin{proof}
  We fix $j$ and focus on the number of $i$-violations in the chain $z^j$.
  Recall that $s_j$ denotes the number of regular blocks in the chain $z^j$
  ($s_j\leq l$).

  If no gadgets have been added during the execution of the algorithm,
  then for all $i\in[0,2k\sqrt{l}]$, the number of $i$-violations is $\leq s_j/2
  \leq l/2$.

  Otherwise, we distinguish on the type of the most frequent violation
  ($i = 0$ or $i>0$).

  \begin{itemize}
  \item Case 1: the most frequent violations are $0$-violations. In that case,
    a proof similar to the case 1 of Lemma~\ref{lemma:fewviolations} (with $0$
    replaced by $a=x^j[0]$ the first letter of $x^j$) shows that the number of
    $i$-violations for any $i\in[0,2k\sqrt{l}]$ is $\leq s_j/2 \leq l/2$.
  \item Case 2: the most frequent violations are $i$-violations for some
    $i>0$. Let us see how the parsing of $0w$ splits the gadgets. As in the
    definition of the gadgets, let $m'=\max(i,m)$ and $v=x^j[0..m-1]1^l$. When
    the first gadget $g_i^0(j) = x^j[0..m'-1]\bar x^j_{m'}$ is added, the red
    parsing splits the gadget $g_i^0(j)$ between $x^j_{<i}$ and
    $x^j[i..m'-1]\bar x^j_{m'}$, because the gadget is added before a regular
    block with an $i$-violation. Furthermore, $x^j[i..m'-1]\bar x^j_{m'}$ is not
    split by the parsing, because at that moment in the algorithm, the number
    of $i$-violations in the previous regular green blocks is $s_j/2 \geq m'-i$,
    so that, as the position $i$ has been seen $\geq m'-i$ times, the word
    $x^j[i..m'-1]$ is already in the dictionary of $0w$. Similarly, the gadget
    $g_i^c(j) = x^j[0..m'-1]\bar x^j_{m'}v[0..c-1]$ is split by the red parsing
    between $x^j_{<i}$ and $x^j[i..m'-1]\bar x^j_{m'}v[0..c-1]$, with the additional
    property that this second part is not split by the parsing.

    But for the gadget $g^{2m+1}_i(j)$, the second part
    $x^j[i..m'-1]\bar x^j_{m'}x^j[0..m-1]1^{m+1}$ is parsed in exactly one
    block because this factor does not appear anywhere in a regular block
    because of $1^{m+1}$ (cf. (P2)) nor in a gadget of a preceding chain
    because of $x^j[0..m-1]$ (cf. (P2) again). From that moment on, each
    $i$-violation creates a $0$-violation. The number of green blocks that are
    both $0$-violated and $i$-violated is at most $i \leq 2k\sqrt{l}$. Thus,
    at most $2k\sqrt{l}$ more gadgets fail to kill the corresponding
    $i$-violation. The total number of ``failing'' gadgets in the chain $z^j$
    is at most $2m + 1 + 2k\sqrt{l}$.
  \end{itemize}
\end{proof}

If the number of $i$-violations is not too large, then the same is true for
the size of offset-$i$ blocks, as the following easy result states.
\begin{lemma}
\label{lemma:violation-size}
If the number of $i$-violations in the chain $z^j$ is $b$, then
any offset-$i$ block $u$ in a regular block of $z^j$ is of size at most
$b + 2k\sqrt{l}$.
\end{lemma}
\begin{proof}
  The $|u|-2k\sqrt{l}$ predecessors of $u$ of size at least $2k\sqrt{l} + 1$ cannot
  appear in gadgets, hence by Property~(P2) they must appear at
  position $i$ in the regular green blocks of the chain $z^j$. Therefore, each
  such predecessor contributes to an $i$-violation in $z^j$, so that
  $|u| - 2k\sqrt{l} \leq b$.
\end{proof}

Now, the next two results show that, for large $i$, the size of offset-$i$
blocks is small. First, we need to bound the number of junction blocks ending
at position $i-1$.
\begin{lemma} 
  \label{lemma:fewjunction1}
  Let $j$ be fixed and $i > 2k\sqrt{l}$. Let $uu'$ be a junction block of
  type $1$ between two regular green blocks $w_{a}^j$ and $w_{a+1}^j$, ending
  at position $i-1$ in $w_{a+1}^j$ (thus $|u'|=i$). Then
  $|u| \leq \log(kl) - \log(i - 4m - 2)$.

  In particular, the number of such blocks is upper bounded
  by $\frac{2kl}{i - 4m - 2}$.
\end{lemma}

\begin{proof}
  Let $v$ be the prefix of size $4m+1$ of $u'$ (which is also the prefix of
  $x^j$). We claim that all the prefixes of $uu'$ of size $\geq |uv|$ are
  junction blocks of type 1 or 3 only (except possibly for one of type 2),
  with only $u$ on the left side of the junction. Indeed, recalling the red
  parsing of gadgets explained in the proof of
  Lemma~\ref{lemma:number-violations} and Property~(P2), we distinguish the
  following cases:
  \begin{enumerate}
  \item $uv$ cannot be completely included in a regular block, otherwise
    $v[0..m-1]$ would appear both at positions $0$ and $p>0$ in $x^j$, which
    contradicts Property~(P2);
  \item $uv$ cannot be completely included in a gadget:
    \begin{itemize}
    \item if the gadget is $g_0^c(j)$, then $v$ would contain $a^{m+1}$ for
      some $a\in\{0,1\}$,
    \item if the gadget is $g_b^c(j)$ for $b>0$, let $m'=\max(b,m)$: the red
      parsing splits this gadget between $x^j[0..b-1]$ and
      $x^j[b..m'-1]\bar x_{m'}^jx^j[0..m-1]1^d$. Then $uv$ is not contained in
      the first part by (P2), nor in the second part since it cannot contain
      $1^{m+1}$;
    \end{itemize}
  \item if $uv$ is a type 1 junction but not split between $u$ and $v$, it is
    impossible because the three possible cases lead to a contradiction:
    \begin{itemize}
    \item if $u$ goes on the right, then $v$ would appear at another position
      $p>0$ in $x^j$,
    \item if $v$ goes on the left by at least $m$, then $v[0..m-1]$ would
      again appear at two different positions in $x^j$,
    \item if $v$ goes on the left by less than $m$, then it goes on the right
      by more than $m$ and $v[3m+1..4m]$ would again appear at two different
      positions in $x^j$;
    \end{itemize}
  \item if $uv$ is a type 2 junction but not split between $u$ and $v$, it is
    again impossible:
    \begin{itemize}
    \item if $u$ goes on the right, then $v$ would appear at another position
      $p>0$ in $x^j$,
    \item if $v$ goes on the left by at least $3m+1$, in case of $g_0^c(j)$
      then $v$ would contain $a^{2m+1}$ and in case of $g_b^c(j)$ then $v$
      would contain $1^{m+1}$ (recall where the red parsing splits this gadget),
    \item otherwise, $v$ goes on the right by at least $m+1$, and
      $v[3m+1..4m]$ would appear at two different positions $x^j$;
    \end{itemize}
  \item if $uv$ is a type 3 junction, first remark that the gadget is of the
    form $g_b^c(j)$ for $b>0$ (because, for gadgets of the form $g_0^c(j)$,
    the red parsing starts at position $0$ of the gadget). If $uv$ is not
    split between $u$ and $v$, it is once again impossible:
    \begin{itemize}
    \item if $u$ goes on the right, the red parsing of the gadget stops after
      $x^j_{<b}$ and $v$ would appear in $x^j$ at a non-zero position,
    \item similarly, if $v$ goes on the left by less than $m$, then
      $v[m..2m-1]$ would again appear at two different positions in $x^j$,
    \item if $v$ goes on the left by at least $m$, then $v[0..m-1]$ would
      again appear at two different positions in $x^j$.
    \end{itemize}
  \end{enumerate}
  Remark finally that all parsings of type 2 junctions have different sizes on
  the left. Therefore, at most one can contain $u$ on the left. The claim is
  proved. Thus $u$ appears at least $i-4m-2$
  times as a suffix of a regular green block.

  Remark that Property~(P1) implies that factors of size more than $k$ appear
  at most $k\sqrt{l}$ times in $x^j$. Thus, since $i-4m-2>k\sqrt{l}$, we have
  $|u|\leq k$. Hence by Property~(P1), the number of occurrences of $u$ is
  upper bounded by $kl/2^{|u|}$. Therefore
  \begin{align*}
    i - 4m - 2 &\leq  \frac{kl}{2^{|u|}}\\
    |u| &\leq \log(kl) - \log(i - 4m - 2),
  \end{align*}
  which proves the first part of the lemma.

  The number of such blocks is then upper bounded by the number of words of
  size $\leq \log(kl) - \log(i-4m -2)$, that is, $\frac{2kl}{i - 4m - 2}$.
\end{proof}

The last lemma completes the preceding one: if an offset-$i$ block is large,
then a lot of blocks have to end at position $i-1$ and too many of their prefixes
would have to be in different green blocks.
\begin{lemma}
\label{lemma:size-offset}
For any $j \in[1,2^p]$ and any $i > 2k\sqrt{l}$, the size of an offset-$i$
block included in a regular green block of the chain $z^j$ is at most
$$2\sqrt{l} + 3k+ 7m + \frac{2kl}{i - 4m - 2}.$$
\end{lemma}

\begin{proof}
  We argue as in Proposition~\ref{prop:smallblocks}. Let $u$ be an offset-$i$
  block included in a regular green block of the chain $z^j$. We show as
  before that the $|u|-3m$ predecessors of $u$ of size $\geq 3m$ have to
  start at position $i$ in regular green blocks. Indeed, let $v$ be a
  prefix of size $\geq 3m$ of $u$; let us analyse the
  different cases:
  \begin{itemize}
  \item If $v$ is included in a regular green block, then it has to start at
    position $i$ by Property~(P2);
  \item $v$ cannot be included in a gadget since it would lead to a contradiction:
    \begin{itemize}
    \item in gadgets of type $g_0^c(j)$, $v$ would contain $a^{m+1}$ for some
      letter $a\in\{0,1\}$,
    \item in gadgets of type $g_b^c(j)$ (for
      $b\in]0,2k\sqrt{l}]$), either $v$ would contain $1^{m+1}$ or a factor of
      $x^j$ of size $m$ and at a position different from $i$;
    \end{itemize}
  \item If $v$ is included in a junction block of type 1, then $v$ starts at
    position $i$ in the left regular block, otherwise either $v[0..m-1]$ would
    be in the left regular block at a position different from $i$, or
    $v[m-1..2m-2]$ would be in the right regular block at a position
    $\leq 3m<i$;
  \item $v$ cannot be included in a junction block of type 2. Indeed, it
    cannot go by $\geq m$ on the right (by~(P2)), thus it goes on the left by
    at least $2m+1$: for $g_0^c(j)$ it would contain $a^{m+1}$ (for some
    $a\in\{0,1\}$), and for $g_b^c(j)$ ($b>0$), it would either contain
    $1^{m+1}$, or a factor of $x^j$ of size $m$ at a position $< m'\leq i$;
  \item If $v$ is included in a junction block of type 3, then $v$ starts at
    position $i$ in the left regular block, otherwise either $v[0..m-1]$ would
    be in the left (regular) block at a position different from $i$, or, by the
    red parsing, the gadget is of type $g_b^c(j)$ (for $b>0$) and $v[m-1..2m-2]$
    would be included in $x^j$ at a position $\leq 3m<i$.
  \end{itemize}
  Thus, at least $|u|-3m$ red blocks end at position $i-1$ in the regular
  blocks of $z^j$. Among them:
  \begin{itemize}
  \item By Lemma~\ref{lemma:fewjunction1}, at most $\frac{2kl}{i - 4m - 2}$
    of them are junctions of type 1.
  \item At most $2m$ of them are junctions of type 2, since from the
    $(2m+1)$-th gadget on, the type 2 junctions end at position $0$ (in case
    of gadgets $g_b^c(j)$ for $b>0$, see the proof of
    Lemma~\ref{lemma:number-violations}) or $\leq m+1\leq i-1$ (in case of
    gadgets $g_0^c(j)$).
  \item There is no junction of type 3 by definition.
  \end{itemize}

  Overall, at least $|u| - 5m - \frac{2kl}{i - 4m - 1}$ of them are
  offset blocks, ending at position $i-1$. We call the set of such
  blocks $M$. 
  Let
  $$A = \frac{|u|-7m-\frac{2kl}{i-4m-2}}{2}$$
  and
  $$S=\{u\in\p(M)\text{ containing } x^j[i-A-2m..i-A-1]\}.$$ We
  say that a red block $w$ is \emph{problematic} if $w\in S$ but the
  part of $w$ corresponding to the factor $x^j[i-A-2m..i-A-m-1]$ is
  not completely included in a regular green block. We show that the
  number of problematic blocks is at most $2k\sqrt{l}+ 2m + 1$.

  \begin{enumerate}
  \item The number of problematic blocks that overlap a gadget
    $g_0^c(j)=ua^c$ in the red parsing is at most $m+1$. Indeed,
    $ua^c$ is never split by the red parsing, therefore for
    $c\geq m+1$, a red block that overlaps $g_0^c(j)$ would contain
    $a^{m+1}$, which is not a factor of $x^j$.
  \item For $b > 0$ (recall that $b\leq 2k\sqrt{l}$), note that the red
    parsing splits the gadget $g_b^c(j)$ after $x^j[0..b-1]$.
    \begin{itemize}
    \item Observe first that the number of problematic blocks that
      overlap the second part of the gadget ($g_b^c(j)_{\geq b}$) is
      at most $m$. Indeed, this part is not split by the red parsing,
      therefore for $c \geq m$ a red block that overlaps this part
      would contain $\bar x^j_{m'}x^j[0..m-1]$, which is not possible
      since the position of the word $x^j[0..m-1]$ should be $0$ by
      Property~(P2)
    \item The number of problematic blocks that appear completely included in
      the first part of a gadget ($g_b^c(j)_{<b}$) is at most $b$. Otherwise,
      the red parsing creates at least one red block completely included in
      the first part $x^j[0..b-1]$ of the gadget, and we claim that this can
      happen at most $b$ times. Indeed, each time the parsing falls in this
      case, the last red block included in the first part of the gadget has to
      end at position $b-1$, but the size of this block has to be different
      each time, so that this second case can occur at most $b$ times.
      Finally, each time a gadget is parsed, at most one of the red blocks
      included in the first part of the gadget can be a word of $S$ by
      Property~(P2).
    \end{itemize}
  \item There is no problematic blocks that are junction blocks of
    type 1 or 3. Indeed, if it were the case, the right part of the
    junction would be of size $\leq m-1$ since otherwise the problematic
    block would contain $ax^j[0..m-1]$ for some letter $a$, which is
    not possible. Therefore, within the problematic block, the factor
    $x^j[i-A-2m..i-A-m-1]$ appears on the left side of the junction
    and is thus included in a regular block.
  \item The number of problematic blocks that are junction blocks of
    type 2 has already been considered when considering the gadgets.
  \end{enumerate}
  All the red blocks corresponding to words of $S$ and that are not
  problematic have to appear in distinct regular green blocks by
  Property~(P2). As before, a word of $S$ is obtained by choosing its
  beginning before the interval and its end after, so that
  $$|S|\geq \bigl(|u|-5m-\frac{2kl}{i-4m-2}-(A+2m)\bigr)\cdot A.$$
  Therefore:
  \begin{align*}
  |S| - (2k\sqrt{l}+ 2m + 1) & \leq l\\
  \left(\frac{|u|-7m-\frac{2kl}{i-4m-2}}{2}\right)^2 -
  (2k\sqrt{l}+ 2m + 1) & \leq l,
  \end{align*}
  so that
  \begin{align*}
  |u| & \leq 2\sqrt{l+ 2k\sqrt{l}+ 2m + 1} + 7m + \frac{2kl}{i -
    4m - 2}\\
    & \leq 2\sqrt{l} + 3k+ 7m + \frac{2kl}{i - 4m - 2}.
  \end{align*}
\end{proof}

\section{Infinite words}
\label{sec:infinite}

The techniques on finite words developed in the preceding sections can almost
be used as a black box to prove the one-bit catastrophe for infinite words
(Theorem~\ref{thm:infinite}). Our aim is to design an infinite word
$w\in\{0,1\}^\nn$ for which the compression ratios of the prefixes tend to
zero, whereas the compression ratios of the prefixes of $0w$ tend to
$\epsilon >0$. In Section~\ref{sec:general}, we concatenated the bricks
obtained in Section~\ref{sec:toy}; now, we concatenate an infinite number of
bricks of Section~\ref{sec:general} of increasing size (with the parameters
that gave the one-bit catastrophe on finite words). As before, each chain of
size $l$ will be parsed in $\Theta(l)$ green blocks and $\Theta(l^{3/2})$ red
blocks. To guarantee that the compression ratio always remains close to zero in
$w$ and never goes close to zero in $0w$, the size of the bricks mentioned
above will be adjusted to grow neither too fast nor too slow, so that the
compression speed will be locally the same everywhere.

We will need an infinite sequence of families $(F_i)_{i\geq 0}$ of words
similar to that of Section~\ref{sec:general}: thus we will need infinite
sequences of parameters to specify them.
\begin{itemize}
\item For $i\geq 0$, the size of words in $F_i$ will be $l_i=l_0.2^i$, for
  $l_0$ sufficiently large.
\item Let $p_i=\sqrt{l_i}/(9\gamma)-2\log l_i$, where $\gamma\geq 10$ is a
  constant. For $i>0$, the number of words in $F_i$ will be
  $|F_i|=2^{p_i}-2^{p_{i-1}}$ (and $|F_0|=2^{p_0}$). Remark that
  $\sum_{j=0}^i|F_j|=2^{p_i}$ and $|F_{i+1}|\sim|F_i|^{\sqrt{2}}$.
\item The parameter $k_i=(\log l_i)/2$ will be the maximal size of words in
  Property~P1$(i)$ below.
\item The parameter $m_i=\gamma p_i$ will be the size of words in
  Property~P2$(i)$ below.
\end{itemize}

We shall later show that there exists an infinite sequence
$\F=(F_i)_{i\geq 0}$ matching these parameters and satisfying some desired
properties (generalized versions of Properties~(P1) and~(P2), see below). But
from an arbitrary sequence $(F_i)_{i\geq 0}$, let us first define the ``base''
word from which $w$ will be constructed.
\begin{definition}\label{def:wf}
  Given a sequence $\F = (F_i)_{i\geq 0}$ where each $F_i$
  is a family of words, we denote by $w_{\F}$ the word
  $$w_{\F}=\prod_{i=0}^{\infty} \prod_{x \in F_i} \pref_{> q_x^i}(x) $$
  where
  $q_x^i = \max \{ a\,:\,x_{< a} \text{ is a prefix of a word
    in }\cup_{j<i} F_j \}$.
\end{definition}
For a particular sequence $\F=(F_i)$, the word $w$ will be equal to $w_\F$
with some gadgets inserted between the prefixes as in the previous sections.
The sequence $\F$ that we shall consider will be a sequence of families of
random words which will satisfy the following properties
(Lemma~\ref{lemma:infinite-family} below shows that these properties are true
with high probability).
\renewcommand{\descriptionlabel}[1]{\hspace{\labelsep}{#1}:} 
\begin{description}
\item[P1$(i)$] For all $x\in F_i$, for all words $u$ of size at most
  $k_i$, $\occ_x(u)\leq k_il_i/2^{|u|}$.
\item[(P1')] For all $i\geq 0$, P1$(i)$.
\item[P2$(i)$] Any factor $u$ of size $m_i$ appears in at most
  one word of $\cup_{j\leq i} F_j$, and within that word at only one position.
\item[(P2')] For all $i\geq 0$, P2$(i)$.
\end{description}
Again, (P2') guarantees a kind of ``independence'' of the families $F_0$,
$F_1$, $\dots$, whereas (P1') is a de Bruijn-style ``local'' property on each
word of each family $F_i$.

Our first lemma shows that there exists a sequence $\F=(F_i)_{i\geq 0}$
satisfying (P1') and (P2').
\begin{lemma}\label{lemma:infinite-family}
  For every $i\geq 0$, let $F_i$ be a set of $2^{p_i}-2^{p_{i-1}}$ words of
  size $l_i$ (and $2^{p_0}$ words of size $l_0$ for $F_0$) taken uniformly and
  independently at random. Then the probability that $\F$ satisfies
  Properties~(P1') and~(P2') is non-zero.
\end{lemma}

\begin{proof}
  Let us show that the probability that $\F$ satisfies (P1') is $>1/2$, and
  similarly for (P2'). We only show it for (P2'), as an analogous (and easier)
  proof gives the result for (P1') as well.

  By Lemma~\ref{lem:p2whp}, the probability that $\F$ does not satisfy
  P2$(i)$ is less than $2/l_i=2^{1-i}/l_0$. Thus, by union bound, the
  probability that all P2$(i)$ are satisfied is larger than
  $$1 -\sum_{i=0}^{\infty} \frac{2^{1-i}}{l_0}=1-\frac{4}{l_0}.$$
\end{proof}

From now on, we consider a sequence of families $\F=(F_i)_{i\geq 0}$, with
parameters $(l_i)$ and $(p_i)$, that has both Properties~(P1') and~(P2') for
the parameters $(m_i)$ and $(k_i)$ defined above. Remark that the integers
$q_x^i$ defined in Definition~\ref{def:wf} satisfy $q_x^i\leq m_i$ thanks to
Property~P2$(i)$.

The word $w$ that we consider is the word $w_{\F}$ (Definition~\ref{def:wf})
where gadgets have possibly been added between the regular green blocks
exactly as in the algorithm of Section~\ref{sec:general}. Since $\F$ satisfies
(P1') and (P2'), and the parameters $l_i$, $p_i$ fall within the range of
Theorem~\ref{thm:lowerbound}, it can be shown as in Section~\ref{sec:general}
that a chain of $w$ coming from $F_i$ will be parsed in $\geq l_i^{3/2}/54$
red blocks in $0w$ but in only $\leq 3l_i/2$ green blocks. The following two
lemmas show Theorem~\ref{thm:infinite}, i.e. that $w$ satisfies the one-bit
catastrophe. We begin with the upper bound on the compression ratio of $w$,
before proving the lower bound for $0w$ in Lemma~\ref{lem:compinf}.
\begin{lemma}\label{lem:compsup}
  $\compsup(w)=0$.
\end{lemma}
\begin{proof}
  By definition (Definition~\ref{def:infinite-compression-ratio}),
  $$\compsup(w)=\limsup_{n\to\infty}\comp(w_{<n}),$$
  therefore we need to show that $\comp(w_{<n})=G(\log G)/n$ tends to zero,
  where $G=|\dic(w_{<n})|$ is the number of green blocks in the parsing of
  $w_{<n}$. Let us evaluate this quantity for a fixed $n$.

  Let $j$ and $q$ be the integers such that the $n$-th bit of $w$ belongs to
  the $q$-th chain of the $j$-th family, or in other terms, that $w_{<n}$ is
  the concatenation of the chains coming from $\cup_{i<j} F_i$ and of the
  first $q-1$ chains of $F_j$, together with a piece of the $q$-th chain of
  $F_j$.

  We first give a lower bound on $n$ as a function of the different
  parameters. A chain coming from Family $F_i$ is of the form
  $\pref_{>q_x^i}(x)$ together with possible gadgets, where $q_x^i\leq
  m_i\leq\sqrt{l_i}$. Therefore, the size of such a chain is at least
  $$\sum_{j=m_i}^{l_i}j\geq \frac{l_i^2-m_i^2}{2}=\frac{l_i^2}{2}-o(l_i^2).$$
  Thus, using $l_j=2l_{j-1}$, we get:
  $$ n + o(n)\geq \sum_{i=0}^{j-1}|F_i|\frac{l_i^2}{2} + (q-1)\frac{l_j^2}{2}\geq
  \frac{l_{j-1}^2}{2}(|F_{j-1}|+q)$$
  as soon as $2(q-1)\geq q/2$, that is, $q\geq 2$. (We shall take care of the
  case $q=1$ below.)

  On the other hand, when all possible gadgets are added, each chain has a
  size at most $5l_i^2/8+o(l_i^2)$ (see Section~\ref{sec:construction}). Using
  the fact that $|F_{i+1}|\sim |F_i|^{\sqrt{2}}$ (i.e. the growth of the
  sequence $(|F_i|)_{i\geq 0}$ is more than exponential), we obtain the
  following upper bound:
  $$n-o(n)\leq \frac{5}{8}\biggl(\sum_{i=0}^{j-1}|F_i|l_i^2 + ql_j^2\biggr)\leq
  \frac{5}{8}(2|F_{j-1}|l_{j-1}^2+ql_j^2)=\frac{5}{4}(|F_{j-1}|+2q)l_{j-1}^2.$$
  In particular,
  $$\log G\leq\log n\leq 2\log|F_{j-1}|.$$

  Let us now bound the number of green  blocks. A chain coming from $F_i$,
  with gadgets, is parsed in at most $3l_i/2$ blocks. Hence
  $$G\leq \frac{3}{2}\biggl(\sum_{i=0}^{j-1}|F_i|l_i+ql_j\biggr)\leq
  \frac{3}{2}(2|F_{j-1}|l_{j-1}+2ql_{j-1})=3l_{j-1}(|F_{j-1}|+q).$$

  We can now bound the compression ratio of $w_{<n}$:
  $$\comp(w_{<n})=\frac{G\log G}{n}\leq
  \frac{6}{l_{j-1}}\cdot 2\log|F_{j-1}|\xrightarrow[n\to\infty]{}0.$$

  Finally, for the case $q=1$, looking back at the inequalities above we have:
  $n+o(n)\geq|F_{j-1}|l_{j-1}^2/2$ and $G\leq 3l_{j-1}(|F_{j-1}|+1)$, thus
  $\comp(w_{<n})$ again tends to zero.
\end{proof}

Finally we turn to the lower bound on the compression ratio of $0w$.
\begin{lemma}\label{lem:compinf}
  $\compinf(0w)\geq 2/(1215\gamma)$.
\end{lemma}

\begin{proof}
  Define $j$ and $q$ as in the proof of Lemma~\ref{lem:compsup}: we want to
  give a lower bound on $(R\log R)/n$, where $R=|\dic(0w_{<n})|$ is the number
  of red blocks in the parsing of $0w_{<n}$. The upper 
  bound for $n$ given there still hold:
  $$n-o(n)\leq \frac{5}{4}(|F_{j-1}|+2q)l_{j-1}^2.$$

  Let us now give a lower bound on $R$. Suppose for now that $q\geq 4$, so
  that $2\sqrt{2}(q-1)\geq 2q$. The proof of Theorem~\ref{thm:lowerbound} in
  Section~\ref{sec:general} shows that each chain coming from a family $F_i$
  is parsed in at least $\epsilon l_i^{3/2}$ red blocks, where
  $\epsilon=1/54$. Hence:
  \begin{align*}
  R&\geq \sum_{i=0}^{j-1} \epsilon |F_i|l_i^{3/2}+\epsilon(q-1)l_j^{3/2}\\
  &\geq \epsilon (|F_{j-1}|l_{j-1}^{3/2}+2\sqrt{2}(q-1)l_{j-1}^{3/2})\\
  &\geq \epsilon (|F_{j-1}|+2q)l_{j-1}^{3/2}.
  \end{align*}
  Therefore,
  $$\frac{R\log R}{n}\geq
  \frac{4\epsilon}{5}\cdot\frac{\log|F_{j-1}|}{\sqrt{l_{j-1}}}\sim
  \frac{4\epsilon}{5\times 9\gamma}.$$

  In the case $q\leq 3$, we have $q<<|F_{j-1}|$ and the same bound holds.
\end{proof}

\section{Future work}
\label{sec:future}

A word on what comes next. As mentioned in the introduction, we have
privileged ``clarity'' over optimality, hence constants can undoubtedly be
improved rather easily. In that direction, a (seemingly harder) question is to
obtain $\compsup(w)=0$ and $\compinf(0w)=1$ in Theorem~\ref{thm:infinite}.

The main challenge though, to our mind, is to remove the gadgets in our
constructions. Remark that the construction of Section~\ref{sec:toy} can also
be performed with high probability with a random word instead of a de Bruijn
sequence (that is what we do in Section~\ref{sec:general} in a more general
way). Thus, if we manage to get rid of the gadgets using the same techniques
presented here, this would mean that the ``weak catastrophe'' is the typical
case for optimally compressible words. Simulations seem to confirm that
conclusion. But, as a hint that removing the gadgets may prove difficult,
Remark~\ref{rem:gadgets} emphasizes the vastly different behaviour of the
LZ-parsings on $1w$ with and without gadgets.

\section{Acknowledgments}
\label{sec:acknowledgments}

We want to thank Elvira Mayordomo for sharing, long time ago, the ``one-bit
catastrophe'' question with us. The long discussions with Sophie Laplante have
helped us present our results in a more readable (less unreadable) way. We
thank Lucas Boczkowski for his close scrutiny of parts of this paper, and
Olivier Carton for useful discussions.

\bibliographystyle{plain}
\bibliography{biblio}
\end{document}